\RequirePackage[2020-02-02]{latexrelease}
\documentclass[manuscript]{biometrika}

\usepackage{amsmath}

\usepackage{times}
\usepackage{bm}
\usepackage{hhline}
\usepackage{natbib}

\graphicspath{{./art/}}

\usepackage[plain,noend]{algorithm2e}

\usepackage{amsmath}
\usepackage{graphicx}
\usepackage{hyperref}
\hypersetup{colorlinks,allcolors=black}

\usepackage{multicol}
\usepackage{multirow}
\usepackage{amsfonts}
\usepackage[shortlabels]{enumitem}

\newcommand{\indep}{\perp \!\!\! \perp}
\newcommand{\expit}{\mathrm{expit}}
\newcommand{\logit}{\mathrm{logit}}

\usepackage{pgf,tikz}
\usetikzlibrary{arrows,shapes.arrows,shapes.geometric,shapes.multipart,decorations.pathmorphing,positioning,shapes.swigs}

\newcommand{\beginsupplement}{
        \setcounter{table}{0}
        \renewcommand{\thetable}{S\arabic{table}}%
        \setcounter{figure}{0}
        \renewcommand{\thefigure}{S\arabic{figure}}%
        \setcounter{section}{0}
        \renewcommand{\thesection}{S\arabic{section}}
        \setcounter{assumption}{0}
        \renewcommand{\theassumption}{S\arabic{assumption}}
        \setcounter{model}{0}
        \renewcommand{\themodel}{S\arabic{model}}
        \setcounter{theorem}{0}
        \renewcommand{\thetheorem}{S\arabic{theorem}}
        \setcounter{lemma}{0}
        \renewcommand{\thelemma}{S\arabic{lemma}}
     }

\makeatletter
\renewcommand{\algocf@captiontext}[2]{#1\algocf@typo. \AlCapFnt{}#2} 
\def\@algocf@capt@plain{top}
\renewcommand{\algocf@makecaption}[2]{%
  \addtolength{\hsize}{\algomargin}%
  \sbox\@tempboxa{\algocf@captiontext{#1}{#2}}%
  \ifdim\wd\@tempboxa >\hsize
    \hskip .5\algomargin%
    \parbox[t]{\hsize}{\algocf@captiontext{#1}{#2}}
  \else%
    \global\@minipagefalse%
    \hbox to\hsize{\box\@tempboxa}
  \fi%
  \addtolength{\hsize}{-\algomargin}%
}
\makeatother

\begin{document}

\jname{Biometrika}

\jyear{0000}
\jvol{000}
\jnum{0}

\markboth{K. Li et~al.}{Biometrika style}

\title{Doubly Robust Proximal Causal Inference under Confounded Outcome-Dependent Sampling}

\author{Kendrick Qijun Li}
\affil{Department of Biostatistics, University of Michigan, \\ 1415 Washington Heights, Ann Arbor, Michigan, U.S.A.
\email{qijunli@umich.edu}}

\author{Xu Shi}
\affil{Department of Biostatistics, University of Michigan, \\1415 Washington Heights, Ann Arbor,  Michigan, U.S.A.
\email{shixu@umich.edu}}

\author{Wang Miao}
\affil{Department of Probability and Statistics, Peking University, \\ 5 Yiheyuan Rd, Haidian District, Beijing, China
\email{mwfy@pku.edu.cn}}

\author{\and Eric Tchetgen Tchetgen}
\affil{Department of Statistics and Data Science, The Wharton School, University of Pennsylvania,\\ 265 South 37th Street, Philadelphia, Pennsylvania, U.S.A
\email{ett@wharton.upenn.edu}}

\maketitle

\begin{abstract}
Unmeasured confounding  and selection bias are often of concern in observational studies and may invalidate a causal analysis if not appropriately accounted for. Under outcome-dependent sampling, a latent factor that has causal effects on the treatment, outcome, and sample selection process may cause both unmeasured confounding and selection bias, rendering standard causal parameters unidentifiable without additional assumptions. Under an odds ratio model for the treatment effect, \citet{li2022double} established both proximal identification and estimation of causal effects by leveraging a pair of negative control variables as proxies of latent factors at the source of both confounding and selection bias. However, their approach relies exclusively on the existence and correct specification of a so-called treatment confounding bridge function, a model that restricts the treatment assignment mechanism. In this article, we propose doubly robust estimation under the odds ratio model with respect to two nuisance functions -- a treatment confounding bridge function and an outcome confounding bridge function that restricts the outcome law, such that our estimator is consistent and asymptotically normal if either bridge function model is correctly specified, without knowing which one is. Thus, our proposed doubly robust estimator is potentially more robust than that of \citet{li2022double}. Our simulations confirm that the proposed proximal  estimators of an odds ratio causal effect can adequately account for both residual confounding and selection bias under stated conditions with well-calibrated confidence intervals in a wide range of scenarios, where standard methods generally fail to be consistent. In addition, the proposed doubly robust estimator is consistent if at least one confounding bridge function is correctly specified.
 \end{abstract}

\begin{keywords}
endogenous selection bias; kernel machine learning; proximal causal inference; semiparametric estimation. 
\end{keywords}

\section{Introduction}

Unmeasured confounding and selection bias are two ubiquitous challenges in observational studies, which may lead to biased causal effect estimates and misleading causal conclusions~\citep{hernan2020causal}. \citet{lipsitch2010negative} reviewed and formalized the use of negative control variables as tools to detect the presence of unmeasured confounding in epidemiological research. Specifically, they identify two types of negative control variables:   negative control exposures (NCE) known a priori to have no causal effect on the outcome and negative control outcomes (NCO) known a priori not to be causally impacted by the  treatment.  Such variables may be viewed as valid proxies of unmeasured confounders to the extent that they are associated with the latter. \citet{tchetgen2014control} and \citet{flanders2017new} proposed regression-type calibration for unmeasured confounding adjustment using NCOs and NCEs, respectively, under fairly strong parametric restrictions. More recently, a double negative control approach which uses NCOs and NCEs jointly to correct for unmeasured confounding has been developed, which achieves nonparametric identification of the average treatment effect~\citep{miao2018identifying}; also see \citep{shi2020selective} for a recent review of negative control methods in Epidemiology.  Methods for identification, estimation and inference about causal parameters by leveraging such proxies to address unmeasured confounding bias have been referred to as ``proximal causal inference''~\citep{tchetgen2020introduction}. 

The literature on proximal causal inference is fast growing. \citet{miao2018confounding} introduced an outcome confounding bridge function approach to identify and estimate the average treatment effect by leveraging both NCE and NCO variables, while \citet{cui2020semiparametric} and \citet{deaner2018proxy} introduced identification via a treatment confounding bridge function. \citet{cui2020semiparametric} further developed a doubly robust approach which can identify and consistently estimate the  average treatment effect if either a treatment or an outcome confounding bridge function exists and can be estimated consistently. \citet{ghassami2021minimax} and \citet{kallus2021causal} concurrently proposed a cross-fitting minimax kernel learning approach for nonparametric estimation of the confounding bridge functions, where the efficient influence function has a mixed bias structure \citep{rotnitzky2021characterization}, such that $\sqrt n$-consistent estimation of the average treatment effect remains possible even when both confounding bridge functions might be estimated at considerably slower than $\sqrt n$ rates.  

Despite the rapid growth in methods to address confounding bias, few have considered  situations in which confounding bias and selection bias might coexist. Selection bias, also sometimes called endogenous selection bias or collider stratification bias~\citep{greenland2003quantifying,elwert2014endogenous,cole2010illustrating,hernan2004structural}, arises when the analysis conditions on selection that may be induced by both the primary treatment and outcome variables.  Selection bias naturally occurs in outcome-dependent sampling designs that are widely used in epidemiological and econometric research to reduce cost and effort, such as case-control design~\citep{breslow1980statistical,pearce2018bias}, case-cohort design~\citep{prentice1986case}, choice-based design~\citep{manski1981alternative}, test-negative design~\citep{jackson2013test,sullivan2016theoretical}, and retrospective cohort designs with electronic health records (EHR) data~\citep{streeter2017adjusting}. Similar to confounding bias, selection bias potentially induces spurious associations between treatment and outcome variables of primary scientific interest, even under the null hypothesis of no causal effect. Existing methods to adjust for selection bias include explicit modeling of sample selection mechanism~\citep{heckman1979sample,heckman1998characterizing,cuddeback2004detecting}, matching~\citep{heckman1996sources,heckman1998characterizing}, difference-in-differences models~\citep{heckman1998characterizing} and inverse probability weighting~\citep{hernan2004structural,mansournia2016inverse}. However, when a common latent factor causes the  treatment, outcome and  selection mechanisms, the above methods are not applicable and the causal effect cannot generally be identified. \citet{gabriel2020causal} referred to such a sampling design as ``confounded outcome-dependent sampling''. They derived causal bounds of the average treatment effect, but these bounds are often too wide to be useful in practice, which highlights the challenges of identification. \citet{didelez2010graphical}, \citet{bareinboim2012controlling} and \citet{bareinboim2014recovering} studied graphical conditions for recoverability of the causal effect encoded on the additive scale and odds ratio scale under different forms of selection bias including outcome-dependent sampling. They concluded that neither causal effect is identifiable if  selection is dependent on both outcome and unmeasured confounders.

Recently, \citet{li2022double} studied inference about an odds ratio model encoding a treatment's association with an outcome of interest conditional on both measured and unmeasured confounders under confounded outcome-dependent sampling, leveraging a pair of negative control variables.  The consistency of their estimator, which we refer to as the proximal inverse probability weighted (PIPW) estimator, requires consistently estimating a treatment confounding bridge function which restricts the treatment assignment mechanism. However, in practice, it may be difficult to posit a suitable parametric model for the treatment confounding bridge function, which may limit the applicability of their approach.

In this paper, we develop semiparametric inference for the conditional odds ratio estimand in~\citet{li2022double}. We present a new identification result and propose the proximal outcome regression (POR) estimator of conditional odds ratio that relies on the existence and correct specification of an outcome confounding bridge function which restricts the outcome law by the conditional odds of the outcome given the treatment and confounders. We further introduce a doubly-robust closed-form expression for the conditional odds ratio, based on which we propose a proximal doubly robust (PDR) estimator, which is consistent if either the treatment or the outcome confounding bridge function is consistently estimated and thus has improved robustness against model misspecification of the nuisance functions compared with PIPW and POR. We demonstrate the performance of the proposed estimators through comprehensive simulations. Throughout, we relegate all proofs to Section~\ref{supp:proof} of the Supplementary Materials.

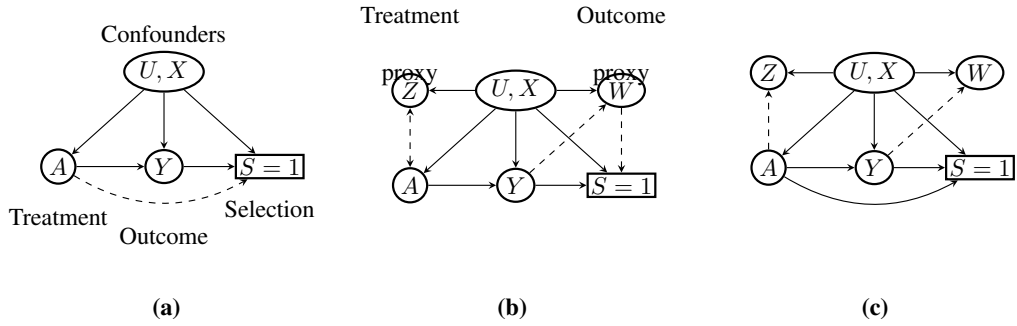
\begin{figure}[!htbp]
		\centering
\resizebox{0.9\textwidth}{!}{
\begin{tabular}{ccc}	
\begin{minipage}{0.31\textwidth}\vskip 1em
	\centering
		\begin{tikzpicture}
			
			\tikzset{line width=1pt,inner sep=5pt,
				swig vsplit={gap=3pt, inner line width right=0.4pt},
				ell/.style={draw, inner sep=1.5pt,line width=1pt}}

			\node[shape = circle, ell] (A) at (-1.67, 0) {$A$};

			\node[shape=ellipse,ell] (Y) at (0,0) {$Y$};
			
			\node[shape=ellipse,ell] (U) at (0,1.5) {$U, X$};
			
			\node[shape=rectangle,ell] (S) at (1.67,0) {$S=1$};

			\node(Aname) at (-1.67, -0.8) {Treatment};
			
			\node(Yname) at (0, -1.1) {Outcome};
			
			\node(Uname) at (0, 2.1) {Confounders};

            \node (Sname) at (1.67, -0.67) {Selection};


			\draw[-stealth, line width=0.5pt, bend right, dashed](A) to (S);
			
			\draw[-stealth,line width=0.5pt](A) to (Y);
			
			\foreach \from/\to in {Y/S, U/S, U/Y, U/A}
			\draw[-stealth, line width = 0.5pt] (\from) -- (\to);
			
		\end{tikzpicture}
	\end{minipage} &
	\hspace{0.1in}
	\begin{minipage}{0.31\textwidth}\vspace{0em}
	\centering
		\begin{tikzpicture}
			
			\tikzset{line width=1pt,inner sep=5pt,
				swig vsplit={gap=3pt, inner line width right=0.4pt},
				ell/.style={draw, inner sep=1.5pt,line width=1pt}}

			\node[shape = circle, ell] (A) at (-1.67, 0) {$A$};
			\node[shape=ellipse,ell] (Y) at (0,0) {$Y$};
			
			\node[shape=ellipse,ell] (U) at (0,1.5) {$U, X$};

			\node[shape=ellipse,ell] (Z) at (-1.67, 1.5) {$Z$};
						
			\node[shape=rectangle,ell] (S) at (1.67,0) {$S=1$};
			
			\node[shape=ellipse,ell] (W) at (1.67, 1.5) {$W$};

			\node (Wname) at (1.67, 2.3) {\begin{tabular}{cc}Outcome\\ proxy\end{tabular}};
			\node (Zname) at (-1.67, 2.3) {\begin{tabular}{cc}Treatment\\ proxy\end{tabular}};

			\node(Yname) at (0, -1.1) {\textcolor{white}{Outcome}};
			
			\draw[-stealth, line width=0.5pt, bend right,color=white](A) to (S);
			
			\draw[-stealth,line width=0.5pt](A) to (Y);
			
			\draw[dashed,stealth-stealth, line width=0.5pt](A) to (Z);
			
			\draw[dashed, -stealth, line width = 0.5pt] (Y) to (W); 
			
			\draw[dashed, -stealth, line width = 0.5pt] (W) to (S);
			
			\foreach \from/\to in {Y/S, U/S, U/Y, U/A, U/W, U/Z}
			\draw[-stealth, line width = 0.5pt] (\from) -- (\to);
			
		\end{tikzpicture}
	\end{minipage}&
	\hspace{0.1in}
	\begin{minipage}{0.31\textwidth}\vskip 1em
	\centering
		\begin{tikzpicture}
			
			\tikzset{line width=1pt,inner sep=5pt,
				swig vsplit={gap=3pt, inner line width right=0.4pt},
				ell/.style={draw, inner sep=1.5pt,line width=1pt}}

			\node[shape = circle, ell] (A) at (-1.67, 0) {$A$};

			\node[shape=ellipse,ell] (Y) at (0,0) {$Y$};
			
			\node[shape=ellipse,ell] (U) at (0,1.5) {$U, X$};
			
			\node[shape=ellipse,ell] (W) at (1.67, 1.5) {$W$};
			
			\node[shape=ellipse,ell] (Z) at (-1.67, 1.5) {$Z$};
			
			\node[shape=rectangle,ell] (S) at (1.67,0) {$S=1$};
			
			\draw[dashed,-stealth, line width=0.5pt](A) to (Z);
			
			\draw[-stealth, line width=0.5pt, bend right](A) to (S);
			
			\draw[-stealth,line width=0.5pt](A) to (Y);
			
			\draw[-stealth,line width=0.5pt, dashed](Y) to (W);
			
			\foreach \from/\to in {Y/S, U/S, U/Y, U/A, U/W, U/Z}
			\draw[-stealth, line width = 0.5pt] (\from) -- (\to);
			
				\node (Wname) at (1.7, 2.1) {\textcolor{white}{Outcome proxy}};
			\node (Zname) at (-1.7, 2.1) {\textcolor{white}{Treatment proxy}};
			\node(Yname) at (0, -1.1) {\textcolor{white}{Outcome}};
		\end{tikzpicture}
	\end{minipage} 
	\\
	\begin{minipage}{0.31\textwidth}\centering$\quad\quad$\textbf{(a)}\end{minipage} &
	\begin{minipage}{0.31\textwidth}\centering$\quad\quad\quad\,\,\,$\textbf{(b)}\end{minipage}
	&\begin{minipage}{0.31\textwidth}\centering$\qquad\quad\quad\quad\,\,\,$\textbf{(c)}\end{minipage}\\
	
\end{tabular}}

\caption{\label{fig:dag} (a) shows the directed acyclic graphs of a confounded outcome-dependent sampling in view. (b) and (c) shows two scenarios where the odds ratio can be identified leveraging additional treatment proxy ($Z$) and outcome proxy ($W$) variables.}
\end{figure}

\section{Identification and estimation under a homogeneous odds ratio model}\label{sec:homo-or}
\subsection{Notation and Setup}\label{sec:setup}
Suppose the data contain a sample of $n$ identically and independently distributed observations drawn from the population of interest, referred to as the ``target population''. For each observation in the target population, let $A$ and $Y$ be the treatment and outcome of primary interest, respectively, assumed to be binary for the time being. Let $S$ be a binary indicator for selection into the study sample, such that the available data only include observations with $S=1$. Let $X$ be a vector of measured pre-treatment covariates that may be associated with $A$, $Y$ and $S$. Suppose that in addition to $X$, there exist pre-treatment latent factors, denoted as $U$, that may cause $A$, $Y$ and $S$. Figure~\ref{fig:dag}(a) shows the directed acyclic graph (DAG) for the causal relationships between the above variables. Similar to \citet{li2022double} Section 2.5, we assume the following model, where $\beta_0$, the conditional log odds ratio given $(U,X)$, encodes the treatment effect of primary interest: 
\begin{model}[Homogeneous odds ratio model; \citet{li2022double} Assumption 3']\label{mod:logit-model}
    \begin{equation*}
    \logit \,P(Y=1\mid A, U, X) = \beta_0 A + \eta(U, X)
\end{equation*}
\end{model}
 where $\logit(x) = \log\left\{x/(1-x)\right\}$ denotes the logit function and $\eta$ is an unknown real-valued function. Model~\ref{mod:logit-model} describes a semiparametric logistic regression model that assumes a homogeneous association between $A$ and $Y$ on the odds ratio scale, across strata of all measured and latent factors $(U,X)$. Model~\ref{mod:logit-model} encodes the structural model
 \begin{equation*}
    \logit \,P(Y(a)=1\mid U, X) = \beta_0 a + \eta(U, X),\qquad{a=0,1},
\end{equation*} where $Y(a)$ ($a=0,1$) denotes the potential outcome were the individual given the treatment status $A=a$, under standard identifiability assumptions of (a) consistency, i.e. $Y(a)=Y$ if $A=a$,  which requires there is only one version of treatment and an individual's outcome is not affected by others' treatment status~\citep{cole2009consistency}; (b) latent ignorability, i.e. $Y(a)\indep A\mid U,X$, which essentially requires $(U,X)$ to contain all confounders of A-Y association; and (c) positivity, i.e. $0<P(A=1\mid U,X)<1$ almost surely~\citep{hernan2020causal}. Under these conditions, $\beta_0$ encodes the log conditional causal odds ratio effect in every $(U,X)$ stratum.
 
 In \citet{li2022double} Assumption 3, they also considered a homogeneous risk ratio model, by which the marginal causal risk ratio can be identified. We focus on Model~\ref{mod:logit-model}  for two reasons: first, proximal identification of $\beta_0$ under Model~\ref{mod:logit-model} permits treatment-induced selection under Assumption~\ref{assump:selection} below, while identification under the homogeneous risk ratio model appears to require that $A$ has no causal effect on selection other than through $Y$; second, \citet{li2022double} invoked a rare outcome assumption in the target population for proximal identification under the homogeneous risk ratio model, which restricts the applicability of the method. Such a rare outcome assumption is not necessary for proximal identification of $\beta_0$ under Model~\ref{mod:logit-model}. As discussed in Section~\ref{supp:or-and-rr} of the Supplementary Materials, however, under standard identifiability assumptions (a)-(c) above and a rare outcome assumption akin to that assumed by \citet{li2022double}, $\beta_0$ in Model~\ref{mod:logit-model} still approximates the log marginal causal risk ratio $P(Y(1)=1) / P(Y(0)=1)$. Alternatively, as discussed in Section~\ref{supp:tnd} of the Supplementary Materials, under settings where control subjects are selected to be positive for an outcome variable that is a priori known to be not affected by the treatment, $\beta_0$ may recover the marginal causal risk ratio without invoking the rare outcome condition. Such a setting may occur, for example, in test-negative design studies of vaccine effectiveness, where the control subjects are selected to have a pre-determined infection that is not affected by the vaccine~\citep{jackson2013test,sullivan2016theoretical,chua2020use,schnitzer2022estimands}.  
 
 Although to simplify the presentation, results given in the main text focus primarily on Model~\ref{mod:logit-model}, Section~\ref{supp:eff-mod} of the Supplemental Materials extends these results to the more general model \begin{equation*}
    \logit \,P(Y=1\mid A, U, X) = \beta_0(X) A + \eta'(U, X),
\end{equation*} which explicitly accounts for effect modification of odds ratio with respect to measured covariates $X$.

In general, the conditional log odds ratio $\beta_0$ is not identifiable: first, the latent factors $U$ defining the strata are not observed; second, the data only include selected subjects with $S=1$ whilst Model~\ref{mod:logit-model} is defined in the target population.

As in \citet{li2022double} Assumption 2', we make an important assumption regarding the selection mechanism that the treatment $A$ does not modify the conditional risk ratio of selection against the outcome in every $(U, X)$ stratum. Formally, we assume:
\begin{assumption}[No effect modification by $A$ on the outcome-dependent selection]\label{assump:selection}For $a=0,1$ and some positive-valued unknown function $\xi$,
$$\dfrac{P(S=1\mid Y=1, A=a, U, X)}{P(S=1\mid Y=0, A=a, U, X)}=\xi(U, X).$$
\end{assumption}

As  mentioned above, a special case of Assumption~\ref{assump:selection} is if $A$ has no causal effect on $S$ other than through $Y$, a stronger assumption considered by \citet{li2022double} to identify the marginal causal risk ratio. As a result of Assumption~\ref{assump:selection}, the conditional odds ratio given $(U, X, S=1)$ is identical to that given $(U, X)$, which leaves only the challenge of controlling for unmeasured confounding by $U$ in the selected sample.

\subsection{Proximal Identification of $\beta_0$ in \citet{li2022double}}\label{sec:PIPW}

To detect and correct for unmeasured confounding bias, the proximal  inference framework proposes to leverage a pair of proxy measurements of unmeasured confounders~\citep{miao2018identifying,cui2020semiparametric,tchetgen2020introduction}. In this line of works, \citet{li2022double} developed proximal causal inference for $\beta_0$ under confounded outcome-dependent sampling as represented in Figure~\ref{fig:dag}(a)~\citep{jackson2013test,chua2020use}. We similarly assume proxies of the latent factors are available.
\begin{assumption}[Proximal independence conditions]\label{assump:nc} There exist a pair of proxies of the latent factor $U$: a treatment proxy (which \citet{li2022double} refers to as a ``negative control exposure''), denoted as $Z$, and an outcome proxy (which \citet{li2022double} refers to as a ``negative control outcome''), denoted as $W$, that satisfy \[\text{(a) } W\indep A\mid U, X, Y, S; \qquad \text{(b) }Z\indep (Y, W, S)\mid A, U, X.\]
\end{assumption}
Assumption~\ref{assump:nc} requires the  treatment proxy $Z$ and outcome proxy $W$ to satisfy certain conditional independences in the target population. Namely, the treatment proxy $Z$ has no direct effect on the primary outcome $Y$, selection indicator $S$, and the outcome proxy $W$; and the primary treatment $A$ has no direct effect on the outcome proxy $W$. It is crucial that both $Z$ and $W$ are both $U$-relevant, that is they carry information and therefore are associated with $U$~\citep{shi2020selective,tchetgen2020introduction}; this is formalized by Assumptions~\ref{assump:completeness-a}(a) and \ref{assump:completeness-b}(a) below.  As such, in the selected sample, an observed $A$-$W$ association after adjusting for $Y, X$, or an observed $Z$-$Y$ or $Z$-$W$ association after adjusting for $A, X$, may indicate the presence of bias due to the latent factor $U$. Figure~\ref{fig:dag} (b) and (c) show the DAGs of scenarios where Assumption~\ref{assump:nc} holds. 

As in \citet{li2022double}, we assumed the existence of a treatment confounding bridge function that connects the treatment proxy to the latent factor.

\begin{assumption}[Treatment confounding bridge function]\label{assump:trt-bridge} There exists a treatment confounding bridge function $q(A, Z, X)$ that satisfies
    \begin{equation}\label{eq:trt-bridge-u}
        E[q(A, Z, X)\mid A=a, U, X, Y=0, S=1] = 1/P(A=a\mid U, X, Y=0, S=1)
    \end{equation}
almost surely for $a=0,1$.
\end{assumption}

  Equation~\eqref{eq:trt-bridge-u} formally defines a Fredholm integral equation of the first kind. To ensure the existence and identifiability of the solution to Equation~\eqref{eq:trt-bridge-u}, we make the following completeness assumptions, similar to those assumed by \citet{li2022double} and \citet{cui2020semiparametric}: 

\begin{assumption}[Completeness]
\label{assump:completeness-a} $\quad$

\begin{enumerate}[(a)]
    \item For any square-integrable function $g$, if $E[g(U)|A, Z, X, Y=0, S=1]=0$ almost surely, then $g(U)=0$ almost surely;
    \item For any square-integrable function $g$, if $E[g(Z)|A, W, X, Y=0, S=1]=0$ almost surely, then $g(Z)=0$ almost surely.
\end{enumerate}
\end{assumption}

The completeness condition has been used to achieve identification for statistical functionals in econometrics and semiparametric statistics literature~\citep{newey2003instrumental,d2011completeness}. Essentially, Assumption~\ref{assump:completeness-a}(a) formalizes the $U$-relevance of $Z$ and requires that $Z$ has at least as much variability as $U$ in every $(A,X)$ stratum of the outcome-free subjects in the selected sample. When $U$ and $Z$ are both categorical variables, $Z$ necessarily has at least as many categories as $U$. Assumption~\ref{assump:completeness-a}(a) and the regularity conditions discussed in Section~\ref{supp:bridge-existence} of the Supplementary Materials, together constitute a sufficient condition for Assumption~\ref{assump:trt-bridge}, i.e. the existence of $q$. 

Similarly, Assumption~\ref{assump:completeness-a}(b) roughly requires that the outcome proxy $W$ is at least as variable as the treatment proxy $Z$, and is a sufficient condition under which $q(A, Z, X)$ can be uniquely identified from the selected sample, as stated in the theorem below:

\begin{theorem}[Identification of $q(A, Z, X)$; \citet{li2022double} Theorem 2']\label{thm:trt-bridge}
Under Assumptions~\ref{assump:nc} and \ref{assump:trt-bridge}, any treatment confounding bridge function $q(A, Z, X)$ satisfies the moment condition
\begin{equation}\label{eq:trt-bridge}
    E\left\{q(a, Z, X)\mid A=a, W, X,Y=0, S=1\right\}=1/P(A=a\mid W, X, Y=0, S=1).
\end{equation}If Assumption~\ref{assump:completeness-a}(b) holds in addition, then the function $q(\cdot)$ is unique and can be identified as the unique solution to Equation~\eqref{eq:trt-bridge}.
\end{theorem}

Leveraging the proxies, \citet{li2022double} established the following identification result for $\beta_0$.

\begin{theorem}[Identification of $\beta_0$; \citet{li2022double} Theorem 1']\label{thm:PIPW}

Under Model~\ref{mod:logit-model}, if Assumptions~\ref{assump:selection}-\ref{assump:trt-bridge} hold, then  the conditional log odds ratio parameter $\beta_0$ solves the moment equation 
\begin{equation}\label{eq:ee-PIPW}
    E\left\{(-1)^{1-A}c(X)Yq(A, Z, X)e^{-\beta A}\mid S=1\right\}=0
\end{equation}
where $c(X)$ is an arbitrary square-integrable unidimensional function that satisfies
$$E\left\{c(X)I(A=a)Yq(A, Z, X)\mid S=1\right\}\neq 0,\quad a=0,1.$$

\end{theorem}

Equation~\eqref{eq:ee-PIPW} admits a closed-form
solution for $\beta_0$:
\begin{equation}\label{eq:beta-PIPW}\beta_0 = \log\left(\dfrac{E\left\{c(X)I(A=1)Yq(A, Z, X)\mid S=1\right\}}{E\left\{c(X)I(A=0)Yq(A, Z, X)\mid S=1\right\}}\right).\end{equation}

Unique identification of $q$ by Assumption~\ref{assump:completeness-a}(b) is necessary. Although Equation~\eqref{eq:beta-PIPW} holds for any treatment confounding bridge function $q$ that satisfies Equation~\eqref{eq:trt-bridge-u}, $q$ can only be identified in the selected sample through solving Equation~\eqref{eq:trt-bridge}. If Equation~\eqref{eq:trt-bridge} has multiple solutions, there is no guarantee that a solution to Equation~\eqref{eq:trt-bridge} is also a solution to Equation~\eqref{eq:trt-bridge-u}.

\subsection{New Proximal Identification Results}
From Theorem~\ref{thm:PIPW}(b), it is clear that the treatment confounding bridge function $q(A,Z,X)$ plays a crucial role for identifying $\beta_0$. However, in case Assumption~\ref{assump:trt-bridge} does not hold, or Assumption~\ref{assump:trt-bridge} holds but Assumption~\ref{assump:completeness-a}(b) does not hold, then the solution to Equation~\eqref{eq:trt-bridge} in the selected sample may not be unique and satisfy  Equation~\eqref{eq:trt-bridge-u} that defines a treatment confounding bridge function. As a result, the parameter $\beta_0$  identified by Equation~\eqref{eq:beta-PIPW}, where $q$ is a solution to Equation~\eqref{eq:trt-bridge}, may not have a causal interpretation. In this section, we establish a new identification result for $\beta_0$ that does not rely on $q(A, Z, X)$. We instead define an outcome confounding bridge function as below:
\begin{assumption}[Outcome confounding bridge function]\label{assump:outcome-bridge}There exists an outcome confounding bridge function $h(A, W, X)$ that satisfies
\begin{equation}
    E\{h(A,W,X)\mid A, U, X,Y=0, S=1\} = \frac{P(Y=1\mid A,U,X,S=1)}{P(Y=0\mid A,U,X, S=1)}\label{eq:h-odds-u}
\end{equation}
almost surely.
\end{assumption}
\citet{miao2018confounding} and \citet{cui2020semiparametric} previously proposed proximal identification of average treatment effect using an outcome confounding bridge function. Our definition of $h(A, W, X)$ by Equation~\eqref{eq:h-odds-u} is different from those in previous works of proximal causal inference, in that the conditional expectation of our $h(A,W,X)$ is defined to equal the conditional odds of the outcome, whilst the conditional expectation of the outcome confounding bridge function in previous works was defined to equal the conditional mean outcome. Theorem~\ref{thm:POR} below suggests that the standard outcome bridge function of prior literature would indeed  fail to identify $\beta_0$ while our alternative definition yields the desired identification result.

 In contrast with Assumption~\ref{assump:completeness-a}, we make the following completeness assumption for the existence and identifiability of the outcome confounding bridge function $h$.

\begin{assumption}[Completeness]
\label{assump:completeness-b} $\quad$

\begin{enumerate}[(a)]
    \item For any square-integrable function $g$, if $E[g(U)|A, W, X, Y=0, S=1]=0$ almost surely, then $g(U)=0$ almost surely;
    \item For any square-integrable function $g$, if $E[g(W)|A, Z, X, Y=0, S=1]=0$ almost surely, then $g(W)=0$ almost surely.
\end{enumerate}
\end{assumption}

 Similar to Assumption~\ref{assump:completeness-a}(a), Assumption~\ref{assump:completeness-b}(a) essentially requires that $W$ has sufficient variability relative to $U$ in every $(A,X)$ stratum of the outcome-free subjects in the selected sample, and constitutes a sufficient condition for Assumption~\ref{assump:outcome-bridge}, i.e. the existence of $h$, together with other regularity conditions, as discussed in Section~\ref{supp:bridge-existence} of the Supplementary Materials. Assumption~\ref{assump:completeness-b}(b) requires the opposite of Assumption~\ref{assump:completeness-a}(b), that the treatment proxy $Z$ is sufficiently variable relative to the outcome proxy $W$, and is a sufficient condition under which $h(A, W, X)$ can be uniquely identified from the selected sample, as stated in Theorem~\ref{thm:outcome-bridge} below.

\begin{theorem}[Identification of $h(A, W, X)$]\label{thm:outcome-bridge}
Under Assumptions~\ref{assump:nc} and \ref{assump:outcome-bridge}, any outcome confounding bridge function $h(A, W, X)$ satisfies the moment condition
    \begin{equation}\label{eq:h-odds}
         E\left\{h(A, W, X)\mid A, Z, X, Y=0, S=1\right\} = \dfrac{P(Y=1\mid A, Z, X, S=1)}{P(Y=0\mid A, Z, X,S=1)}.
    \end{equation}
    If Assumption~\ref{assump:completeness-b}(b) holds in addition, then the function $h(\cdot)$ is unique and can be identified as the unique solution to Equation~\eqref{eq:h-odds}.
\end{theorem}

We introduce the new identification result in Theorem~\ref{thm:POR} below:

\begin{theorem}[Identification of $\beta_0$]\label{thm:POR}

Under Model~\ref{mod:logit-model}, if Assumptions~\ref{assump:selection}, \ref{assump:nc} and \ref{assump:outcome-bridge} hold, then  the conditional log odds ratio parameter $\beta_0$ solves the moment equation 
\begin{equation}\label{eq:ee-POR}
    E\left[c(X)(1-Y)\left\{h(1, W, X)e^{-\beta_0} - h(0, W, X)\right\}\mid S=1\right]=0;
\end{equation}
where $c(X)$ is an arbitrary square-integrable unidimensional function that satisfies
$$E\left\{c(X)(1-Y)h(a, W, X)\mid S=1\right\}\neq 0,\quad a=0,1.$$

\end{theorem}

The integral equation~\eqref{eq:ee-POR} also admits a closed-form solution

\begin{equation}\label{eq:beta-POR}\beta_0 = \log\left(\dfrac{E\left\{c(X)(1-Y)h(1, W, X)\mid S=1\right\}}{E\left\{c(X)(1-Y)h(0, W, X)\mid S=1\right\}}\right).\end{equation}

\subsection{A doubly robust closed-form expression for $\beta_0$}

Until now, identification of  $\beta_0$ has required the existence and identification of either a treatment confounding bridge function or an outcome confounding bridge function. 
In this section, we present a closed-form expression for $\beta_0$, which has a desirable doubly-robustness property in the sense that it only requires one of the two confounding bridge functions to exist and be identified.  
\begin{theorem}[A doubly-robust closed-form expression for $\beta_0$]\label{thm:PDR}
 Under Model~\ref{mod:logit-model} and Assumptions~\ref{assump:selection} and \ref{assump:nc}, we have

\begin{align}
\beta_0 &= \log\bigg(E\bigg\{c(X)\bigg[I(A=1)q(A, Z, X)\left\{Y - (1-Y)h(A, W, X)\right\}+\nonumber\\
&\qquad (1-Y)h(1, W, X)\bigg]\mid S=1\bigg\} /E\bigg\{c(X)\bigg[I(A=0)q(A, Z, X)\{Y - \nonumber\\
&\qquad (1-Y)h(A, W, X)\}+(1-Y)h(0, W, X)\bigg]\mid S=1\bigg\}\bigg),\label{eq:beta-PDR}
\end{align}
if either $q$ is a solution to Equation~\eqref{eq:trt-bridge-u}, or $h$ is a solution to Equation~\eqref{eq:h-odds-u}. Here $c(X)$ is an arbitrary square-integrable real-valued function satisfying
\begin{align*}
    &E\bigg\{c(X)\bigg[I(A=a)q(A, Z, X)\{Y -(1-Y)h(a, W, X)\}+\\
&\qquad (1-Y)h(a, W, X)\bigg]\mid S=1\bigg\}\neq 0
\end{align*}
for $a = 0,1$.
\end{theorem}
In Equation~\eqref{eq:beta-PDR}, setting $h(A,Z,X)=0$ recovers the closed-form solution~\eqref{eq:beta-PIPW}, while setting $q(A, Z, X)=0$ gives the closed-form solution~\eqref{eq:beta-POR}, illustrating the double robustness property of the closed-form expression~\eqref{eq:beta-PDR}.

In Section~\ref{supp:assump} of the Supplementary Materials, we compare our assumptions and results with those in~\citet{li2022double}, and summarise the additional conditions under which $\beta_0$ identifies the marginal causal risk ratio.

\subsection{Estimation and large sample inference}\label{sec:estimation}

Denote data for the $i$th subject in the selected sample as $O_i=(A_i, Y_i, Z_i, W_i, X_i)$, $i=1,\dots, n$. As suggested by Equations~\eqref{eq:beta-PIPW}, \eqref{eq:beta-POR} or \eqref{eq:beta-PDR}, consistent estimators for $q(A, Z, X)$ and $h(A, W, X)$ can be used to construct estimators for $\beta_{0}$ following the standard plug-in principle~\citep{bickel2003nonparametric}. It remains to estimate the two confounding bridge functions. Below, we present crucial moment conditions for the identification of the confounding bridge functions.

\begin{theorem}[Moment conditions for $q(A, Z, X)$ and $h(A, W, X)$]\label{thm:q-h-ident}
Under Assumptions~\ref{assump:nc}, \ref{assump:trt-bridge} and \ref{assump:outcome-bridge}, the confounding bridge functions $q(\cdot)$ and $h(\cdot)$ satisfy 
\begin{align}
    &E\left[(1-Y)\{\kappa_1(A, W, X)q(A, Z, X) - \kappa_1(1, W, X) - \kappa_1(0, W, X))\}\mid S=1\right]=0\label{eq:q-identification}\\
    &E\left[\kappa_2(A, Z, X)\{(1-Y)h(A, W, X) - Y\}\mid S=1\right]=0\label{eq:h-identification}
\end{align}
where $\kappa_1(\cdot)$ and $\kappa_2(\cdot)$ are arbitrary square-integrable functions.
\end{theorem}

Equations~\eqref{eq:beta-PIPW}, \eqref{eq:beta-POR} and \eqref{eq:beta-PDR} and Theorem~\ref{thm:q-h-ident} together suggest a parametric approach to estimate $\beta_0$: one may postulate suitable parametric working models $q(A, Z, X;\tau)$ and $h(A, W, X;\psi)$ where $\tau$ and $\psi$ are unknown parameters of finite dimensions. One can then estimate the nuisance parameters $\tau$ and $\psi$ by solving the estimating equations
\begin{align}
    &\dfrac{1}{n}\sum_{i=1}^n (1-Y_i)\left\{\kappa_1(A_i, W_i, X_i)q(A_i, Z_i, X_i;\tau)-\kappa_1(1, W, X)-\kappa_1(0, W, X)\right\}=0,\label{eq:q-estimation}\\
    &\dfrac{1}{n}\sum_{i=1}^n \kappa_2(A_i, Z_i, X_i)\left\{(1-Y_i)h(A_i, W_i, X_i;\psi)-Y_i\right\}=0.\label{eq:h-estimation}
\end{align}
with user-specified functions $\kappa_1$ and $\kappa_2$ of dimensions equal to that of $\tau$ and $\psi$ respectively.  Closed-form parametric models for the confounding bridge functions can be derived in some settings as described below.

\begin{example}\label{eg:bin}If $U, Z, W$ are all binary variable and $X$ is a categorical variable with levels $\{x_1, \dots, x_{L}\}$, then a suitable model is the saturated model of the form
\begin{align}
    &q(a,z,x;\tau)=\tau_{0x} + \tau_{1x}a + \tau_{2x}z + \tau_{3x}az\label{eq:q-saturated}\\ \text{and}\qquad  &h(a,w,x;\psi)=\psi_{0x} + \psi_{1x}a + \psi_{2x}w + \psi_{3x}aw\label{eq:h-saturated}
\end{align}
for $x=\{x_1, \dots, x_{L}\}$, where\begin{align*}
\tau &= (\tau_{0x_1}, \tau_{1x_1},\tau_{2x_1}, \tau_{3x_1}, \dots, \tau_{0x_L}, \tau_{1x_L},\tau_{2x_L}, \tau_{3x_L})^T,\\
\text{and}\qquad \psi &= (\psi_{0x_1}, \psi_{1x_1},\psi_{2x_1}, \psi_{3x_1}, \dots, \psi_{0x_L}, \psi_{1x_L},\psi_{2x_L}, \psi_{3x_L})^T.
\end{align*}

\end{example}

\begin{example}\label{eg:cont} If $Z$ and $W$ are both multivariate Gaussian random variables conditioning on $(A, U, X, Y)$, under the setting given in Section~\ref{supp:bridge-example-cont} of the Supplementary Materials, suitable models for the confounding bridge functions are 
\begin{align}
&     q(a, z, x;\tau) = 1+\exp\{(1-2A)(\tau_0+ \tau_{A}a + \tau_Z^Tz+\tau_{X}^Tx)\}\label{eq:q-parametric-cont}
    \\
    \text{and}\qquad & h(a, w, x;\psi) = \exp(\psi_0 + \psi_{A}a + \psi_W^Tw + \psi_X^Tx)\label{eq:h-parametric-cont}.
\end{align}
where $\tau = (\tau_0, \tau_A, \tau_Z^T,\tau_X^T)^T$ and $\psi= (\psi_0, \psi_A, \psi_Z^T,\psi_X^T)^T$.
\end{example}

After obtaining the estimators $\hat\tau$ and $\hat\psi$ by solving Equations~\eqref{eq:q-estimation} and \eqref{eq:h-estimation}, $\beta_0$ can then be estimated with the following plug-in estimators:
\begin{align}
    \hat\beta_{c,\text{PIPW}}&=\log\left(\dfrac{\frac{1}{n}\sum_{i=1}^n I(A_i=1)c(X_i)q(1, Z_i, X_i;\hat\tau)Y_i}{\frac{1}{n}\sum_{i=1}^n I(A_i=0)c(X_i)q(0, Z_i, X_i;\hat\tau)Y_i}\right);\label{eq:PIPW}\\
    \hat\beta_{c,\text{POR}}&=\log\left(\dfrac{\frac{1}{n}\sum_{i=1}^n c(X_i)(1-Y_i)h(1, W_i, X_i;\hat\psi)}{\frac{1}{n}\sum_{i=1}^n c(X_i)(1-Y_i)h(0, W_i, X_i;\hat\psi)}\right);\label{eq:POR}\\
    \hat\beta_{c,\text{PDR}} &= \log\bigg(\dfrac{1}{n}\sum_{i=1}^n\bigg\{c(X_i)\bigg[I(A_i=1)q(A_i, Z_i, X_i)\left\{Y_i - (1-Y_i)h(A_i, W_i, X_i)\right\}+\nonumber\\
&\qquad (1-Y_i)h(1, W_i, X_i)\bigg]\bigg\} \bigg/\dfrac{1}{n}\sum_{i=1}^n\bigg\{c(X_i)\bigg[I(A_i=0)q(A_i, Z_i, X_i)\{Y_i - \nonumber\\
&\qquad (1-Y_i)h(A_i, W_i, X_i)\}+(1-Y_i)h(0, W_i, X_i)\bigg]\bigg\}\bigg).\label{eq:PDR}
\end{align}
Alternatively, one can jointly estimate $\beta_0$ and nuisance parameters $\tau$ and $\psi$ by solving estimating equations with the following estimating functions:

\begin{align*}
&G_{\text{PIPW}}(O_i;\beta, \tau) =\\ 
&\quad \left(\begin{array}{l} 
(1-Y_i)\left\{\kappa_1(A_i, W_i, X_i)q(A_i, Z_i, X_i;\tau)-\kappa_1(1, W, X)-\kappa_1(0, W, X)\right\}\\
c(X_i){(-1)}^{1-A_i}q(A_i, Z_i, X_i; \tau)Y_ie^{-\beta A_i}\end{array}\right),\\
&G_{\text{POR}}(O_i;\beta, \psi) =\\ 
&\quad \left(\begin{array}{l} 
\kappa_2(A_i, Z_i, X_i)\left\{(1-Y_i)h(A_i, W_i, X_i;\psi)-Y_i\right\}\\
c(X_i)(1-Y_i)\left\{h(1, W_i, X_i;\psi)e^{-\beta} - h(0, W_i, X_i;\psi)\right\}\end{array}\right),\\
&\text{and}\\
&G_{\text{PDR}}(O_i;\beta, \tau,\psi)=\\
&\quad \left(\begin{array}{l}(1-Y_i)\left\{\kappa_1(A_i, W_i, X_i)q(A_i, Z_i, X_i;\tau)-\kappa_1(1, W, X)-\kappa_1(0, W, X)\right\}\\
\kappa_2(A_i, Z_i, X_i)\left\{(1-Y_i)h(A_i, W_i, X_i;\psi)-Y_i\right\}\\
c(X_i)\bigg[(-1)^{1-A_i}q(A_i, Z_i, X_i;\tau)e^{-\beta A}\left\{Y_i - (1-Y_i)h(A_i, W_i, X_i;\psi)\right\}+\\
        \qquad \qquad (1-Y_i)\left\{h(1, W_i, X_i;\psi)e^{-\beta} - h(0, W_i, X_i;\psi)\right\}\bigg]\end{array}\right).
\end{align*}

The resulting estimators are consistent and asymptotically normal, following standard estimating equation theory~\citep{van2000asymptotic}. Similar to~\citet{cui2020semiparametric}, we refer to $\hat\beta_{c,\text{PIPW}}$, $\hat\beta_{c,\text{POR}}$ and $\hat\beta_{c,\text{PDR}}$ as the proximal inverse probability weighted (PIPW) estimator, the proximal outcome regression (POR) estimator, and the proximal doubly robust (PDR) estimator, respectively. As discussed in Sections~\ref{supp:or-and-rr} and \ref{supp:tnd} in the Supplementary Materials, the PIPW estimator can estimate the marginal causal risk ratio assuming either (a) a homogeneous risk model, a treatment-independent selection mechanism, and a rare outcome condition in the target population, or (b) a test-negative design where the controls are selected to have an irrelevant outcome to the treatment. Because they estimate the same functional, the new POR and PDR estimators can be viewed as estimators of the marginal causal risk ratio under the same conditions.

As implied by Theorem~\ref{thm:PDR}, the estimator $\hat\beta_{c,\text{PDR}}$ also enjoys the desirable double robustness property-- it is consistent for $\beta_{0}$ if either $q(A, Z, X)$ or $h(A, W, X)$ can be consistently estimated. We stated this result formally in the theorem below.

\begin{theorem}\label{thm:doubly-robust}
Under Model~\ref{mod:logit-model} and Assumptions~\ref{assump:selection}, \ref{assump:nc}, \ref{assump:trt-bridge}, and \ref{assump:outcome-bridge},  we have the following results:
\begin{enumerate}[(a)]
    \item $\hat\beta_{c,\text{PIPW}}$ is consistent for $\beta_0$ and asymptotically normal if the model $q(A, Z, X;\tau)$ is correctly specified and Assumption~\ref{assump:completeness-a}(b) holds;
    \item $\hat\beta_{c,\text{POR}}$ is consistent for $\beta_0$ and asymptotically normal if the model $h(A, W, X;\psi)$ is correctly specified and Assumption~\ref{assump:completeness-b}(b) holds;
    \item $\hat\beta_{c,\text{PDR}}$ is consistent for $\beta_0$ and asymtotically normal if either the conditions in (a) or those in (b) hold. 
\end{enumerate}
\end{theorem}

Theorem~\ref{thm:doubly-robust} indicates that the proposed PDR estimator has improved robustness over the PIPW and POR estimator against model misspecification of the two confounding bridge functions.  If candidate proxies $W$ and $Z$ have different cardinalities, for example, when they are categorical variables of unequal dimensions, then only one of the completeness Assumptions~\ref{assump:completeness-a} and \ref{assump:completeness-b} may hold. Say $W$ is of higher dimension, then Assumption~\ref{assump:completeness-b}(b) cannot hold, the outcome bridge function $h(A, W, X)$ may not be not uniquely identified, and $\hat\beta_{c,\text{POR}}$ may be biased for $\beta_0$. However, it may be possible to coarsen levels of $W$ to match those of $Z$~\citep{shi2020multiply}. We highlight that the consistency of $\hat\beta_{c,\text{PDR}}$ only requires one of the two completeness assumptions to hold.

Our identification and estimation strategies involve user-specified functions $c(\cdot)$, $\kappa_1(\cdot)$ and $\kappa_2(\cdot)$. In principle, these nuisance functions can be chosen to construct a potentially more efficient estimator for $\beta_0$, $q$, and $h$ respectively, but the optimal choice of these functions  may require modeling complex features of the observed data distribution. As discussed in \citet{cui2020semiparametric} and \citet{stephens2014locally}, such features can be considerably difficult to model correctly, and unsuccessful attempts to do so may lead to loss of efficiency, thus are not always worthwhile. In our simulation and real data analysis, we found that our generic choices for $c(\cdot)$, $\kappa_1(\cdot)$ and $\kappa_2(\cdot)$ deliver reasonable efficiency.

Until now, our estimation has required specification of parametric working models for $q(A, Z, X;\tau)$ and/or $h(A, W, X;\psi)$. When $X$ is continuous or high-dimensional, however, neither of these working models may contain the true confounding bridge functions. \citet{ghassami2021minimax} and \citet{kallus2021causal} concurrently proposed a cross-fitting minimax kernel learning method for proximal learning of average treatment effect. Their method allows nonparametric estimation of the nuisance functions and  $\sqrt n$-consistent estimation for the average treatment effect, even when both estimated confounding bridge functions are consistent at slower than $\sqrt n$ rates. In Section~\ref{supp:RKHS} of the Supplementary Materials, we demonstrate that their method is also applicable to our setting, and describe a semiparametric kernel estimator for the odds ratio parameter with flexibly estimated confounding bridge functions, $\sqrt n$-consistency, and asymptotic normality.

\section{Simulation}\label{sec:sim}

In this section, we perform simulation studies to evaluate the performance of our proposed estimators of $\beta_{0}$ under five different scenarios. All scenarios follow Model~\ref{mod:logit-model} with $\beta_0=\log(0.2)$.  In Scenario I, the variables $(U, Z, W)$ are all univariate binary variables. We suppose there are no measured confounders $X$. Therefore, saturated models for the confounding bridge functions are appropriate, i.e.
\begin{align*}
    q(A, Z;\tau) = \tau_0 + \tau_aA + \tau_z Z + \tau_{az}AZ \text{ and }
    h(A, W;\psi) = \psi_0 + \psi_a A + \psi_z Z + \psi_{az}AZ.
\end{align*}
The nuisance parameters $\tau$ and $\psi$ are estimated by solving Equations~\eqref{eq:q-estimation} and \eqref{eq:h-estimation}, respectively, with $\kappa_1(A, W)=(1, A, W, AW)^T$ and $\kappa_2(A, Z) = (1, A, Z, AZ)^T$. In Scenario II, the variables $(X, U, Z,W)$ are continuous. The conditional distributions of $Z$ and $W$ given other variables follow Example~\ref{eg:cont}. Therefore parametric working models in Equations~\eqref{eq:q-parametric-cont} and \eqref{eq:h-parametric-cont} are appropriate. We use these parametric working models and estimate the nuisance parameters $\tau$ and $\psi$ by solving estimating equations~\eqref{eq:q-estimation} and \eqref{eq:h-estimation}, setting $\kappa_1(A, W, X)$ to be a vector including $A$, $W$, $X$, all of their higher-order interactions, and an intercept term, and $\kappa_2(A, Z, X)$ similarly. We set $c(X)=1$ for all the estimators under comparison.

In Scenarios III-V, the data generating process is the same as in Scenario II but we misspecify different components of the parametric working models. In Scenario III, we misspecify the treatment confounding bridge function $q(A, Z, X)$ by ignoring $X$, i.e., we set
$$q(a, z, x;\tau) = 1+\exp\{(1-2A)(\tau_0 + \tau_A a + \tau_Zz)\}.$$
In this scenario, we expect the POR and PDR estimators to be consistent and the PIPW estimator to be inconsistent. In Scenario IV, we similarly misspecify the outcome confounding bridge function $h(A,W,X)$ by ignoring $X$, i.e. we set
$$h(a, w, x;\psi)=\exp(\psi_0 + \psi_A a + \psi_w w).$$
In this scenario, we expect the PIPW and PDR estimators to be consistent and the POR estimator to be inconsistent. Finally, in Scenario V, we misspecify both $q(A, Z, X)$ and $h(A, W, X)$ as above. We expect all three estimators to be inconsistent.

We relegate the details of the data generating process to Section~\ref{supp:sim} of the Supplementary Materials. We report the bias, standard deviation, and coverage of 95\% confidence intervals of the POR, PIPW, and PDR estimators for each scenario over 500 Monte Carlo samples.

Table~\ref{tab:sim_results} shows the results of various simulations. In Scenario I and II, as the parametric working models for both confounding bridge functions are correctly specified, all three estimators have essentially identical performances as expected. In Scenario III, the POR and PDR estimators are consistent and have calibrated 95\% confidence intervals, while PIPW is inconsistent and has anti-conservative 95\% confidence intervals. On the other hand, in Scenario IV, the PIPW and PDR estimators are consistent and have calibrated 95\% confidence intervals but not the POR estimator. In Scenario V, all three estimators are inconsistent and have anti-conservative 95\% confidence intervals, but the PDR estimator appears to have slightly smaller biases and better confidence interval coverage.

\begin{table}[!htbp]
    \centering
        \caption{Bias, standard deviation, and coverage of 95\% confidence intervals of the proximal outcome regression (POR), proximal inverse probability weighted (PIPW), and proximal doubly robust (PDR) estimator for five different scenarios summarized over 500 Monte Carlo samples. In Scenario I, the variables $(U,Z,W)$ are univariate binary; in Scenario II-V, the variables $(X,U,Z,W)$ are univariate continuous variables.}
        \resizebox{0.98\textwidth}{!}{
    \begin{tabular}{ll|ccc|ccc|ccc}
    \hline
    \multirow{2}{*}{Scenario} & \multirow{2}{*}{N} &  \multicolumn{3}{c|}{POR} & \multicolumn{3}{c|}{PIPW} & \multicolumn{3}{c}{PDR}\\

   &  & Bias & SD & Coverage & Bias & SD &  Coverage & Bias & SD &  Coverage\\
     \hline 
       I & 3,500 &  0.05 & 0.27 & 97.4\% & 0.05 & 0.05 & 97.4\% & 0.05 & 0.27 & 97.4\% \\
       I & 5,000 &  0.00 & 0.23 & 96.4\% & 0.02 & 0.22 & 96.4\% & 0.00 & 0.22 & 96.4\%\\
       &&&&&&&&&&\\
              II & $1\times 10^5$ &  -0.03 & 0.24 & 95.0\% & -0.03 & 0.23 & 94.2\% & -0.02 & 0.23 & 94.2\% \\
       II & $2\times 10^5$ &  -0.02 & 0.16 & 95.0\% & -0.02 & 0.15 & 94.8\% & -0.02 & 0.15 & 94.2\% \\
       &&&&&&&&&&\\
       III & $1\times 10^5$ &  -0.03 & 0.24 & 94.8\% & -0.11 & 0.22 & 91.2\% & -0.02 & 0.23 & 94.2\% \\
       III & $2\times 10^5$ &  -0.02 & 0.16 & 95.0\% & -0.10 & 0.15 & 90.0\% & -0.02 & 0.15 & 95.0\% \\
       &&&&&&&&&&\\
       IV & $1\times 10^5$ &  -0.12 & 0.24 & 92.2\% & -0.03 & 0.23 & 94.0\% & -0.03 & 0.23 & 94.0\% \\
       IV & $2\times 10^5$ &  -0.11 & 0.16 & 90.8\% & -0.02 & 0.15 & 94.6\% & -0.02 & 0.15 & 94.2\% \\
       &&&&&&&&&&\\
       V & $1\times 10^5$ &  -0.12 & 0.24 & 92.2\% & -0.11 & 0.22 & 90.8\% & -0.09 & 0.22 & 92.2\% \\
       V & $2\times 10^5$ &  -0.11 & 0.16 & 90.8\% & -0.10 & 0.15 & 89.8\% & -0.08 & 0.15 & 91.6\% \\
       \hline
    \end{tabular}}

    \label{tab:sim_results}
\end{table}

\section{Extensions}\label{sec:discussion}

 We have focused on the settings where the treatment and outcome are both binary, but the proposed framework can extend to study the association between a polytomous treatment and a polytomous outcome with simple modification. These extensions are relegated to Section~\ref{supp:polytomous} of the Supplementary Materials.
 
 Although we have primarily focused on test-negative design as an example of outcome-dependent sampling, our method applies to other outcome-dependent sampling designs where unmeasured confounding is of concern and suitable proxy variables may be identified, such as a case-cohort study~\citep{breslow1980statistical} or a retrospective case-control/cohort study using data from electronic health records~\citep{streeter2017adjusting}, where subjects' treatment status may be associated with underlying frailty or risk of developing the outcome of interest. 


\bibliographystyle{biometrika}
\bibliography{references.bib}
\clearpage

\beginsupplement
\setcounter{page}{1}
\begin{center}
    {\textbf{\large Supplementary Material to ``Proximal learning of odds ratio under confounded outcome-dependent sampling designs''}}
\end{center}

\section{Proofs of Theorems.}
\label{supp:proof}

\subsection{Proof of Theorem~\ref{thm:trt-bridge}}

To prove Equation~\eqref{eq:trt-bridge}, we have 
\begin{align}
    & E[ q(A, Z, X)\mid A, W, X, Y=0, S=1]\nonumber\\
    =& E\left\{E[q(A, Z, X)\mid A, U, W, X, Y=0, S=1]\mid A, W, X, Y=0, S=1\right\}\nonumber\\
    \stackrel{A.\ref{assump:nc}}{=}& E\left\{E[q(A, Z, X)\mid A, U, X, Y=0, S=1]\mid A, W, X, Y=0, S=1\right\}\nonumber\\
    \stackrel{A.\ref{assump:nc}}{=}& E\left\{E[ q(A, Z, X)\mid A, U, X]\mid A, W, X, Y=0, S=1\right\}\nonumber\\
    \stackrel{A.\ref{assump:trt-bridge}}{=}&E\left\{\dfrac{1}{P(A\mid  U, X, Y=0, S=1)}\mid A, W, X, Y=0, S=1\right\}\nonumber\\
    =& \int \dfrac{1}{P(A\mid  U=u, X, Y=0, S=1)}f(u\mid A, W, X, Y=0, S=1)du \nonumber\\
    =& \int \dfrac{1}{P(A\mid  U=u, X, Y=0, S=1)}\times\\&\qquad \dfrac{f(u\mid  W, X, Y=0, S=1)P(A\mid U=u, W, X, Y=0, S=1)}{P(A\mid W, X, Y=0, S=1)}du\nonumber\\
    \stackrel{A.\ref{assump:nc}}{=}& \int \dfrac{1}{P(A\mid  U=u, X, Y=0, S=1)}\times\\&\qquad \dfrac{f(u\mid  W, X, Y=0, S=1)P(A\mid U=u, X, Y=0, S=1)}{P(A\mid W, X, Y=0, S=1)}du\nonumber\\
    =& \dfrac{1}{P(A\mid W, X, Y=0, S=1)}\int f(u\mid  W, X, Y=0, S=1)du\nonumber\\
    =& \dfrac{1}{P(A\mid W, X, Y=0, S=1)}.\nonumber
\end{align}

To prove the uniqueness of a solution to Equation~\eqref{eq:trt-bridge} under Assumption~\ref{assump:completeness-a}(b), we suppose there exists another function $q^*$ that solves Equation~\eqref{eq:trt-bridge}, that is,
$$E\{q^*(A, Z, X)\mid A, W, X, Y=0, S=1\} = 1/P(A\mid W, X, Y=0, S=1).$$
Then
$$E\{q^*(A, Z, X)-q(A, Z, X)\mid A, W, X, Y=0, S=1\} = 0.$$

By Assumption~\ref{assump:completeness-a}(b), $q^*(A, Z, X)=q(A, Z, X)$ almost surely.

Suppose there exists a function $q_0$ that solves Equation~\eqref{eq:trt-bridge-u}, then $q_0$ is also the solution to~\eqref{eq:trt-bridge} by the proof above. By the uniqueness of a solution to Equation~\eqref{eq:trt-bridge}, $q_0(A, Z, X)=q(A,Z,X)$ almost surely, which proves that $q$ is the unique solution to Equation \eqref{eq:trt-bridge-u}.

\clearpage
\subsection{Proof of Theorem~\ref{thm:PIPW}}

We first state a useful result for identification:

\begin{lemma}\label{lemma:estimand}Under Model~\ref{mod:logit-model} and a selection mechanism that satisfies Assumption~\ref{assump:selection}, the log odds ratio satisfies 
\begin{equation}\label{eq:or-equation}
    \beta_0 = \log\left(\theta_{1c}/\theta_{0c}\right),
\end{equation}
where, for $a=0,1$, 
\begin{equation}\label{eq:estimand}
    \theta_{ac} = E\left\{c(X)\dfrac{P(A=a\mid U, X,Y=1, S=1)}{P(A=a\mid U, X, Y=0, S=1)}P(Y=1\mid U, X, S=1)\mid  S=1\right\}.
\end{equation}
\end{lemma}

\begin{proof}
Under Model~\eqref{mod:logit-model} and Assumption~\ref{assump:selection}, we have
\begin{align*}
    &\dfrac{P(Y=1\mid U, X, A=1, S=1)P(Y=0\mid U, X, A=0, S=1)}{P(Y=1\mid U, X, A=0, S=1)P(Y=0\mid U, X, A=1, S=1)} \\=& \exp(\beta_0)\times \dfrac{P(S=1\mid Y=1, A=1, U, X)}{P(S=1\mid Y=0, A=1, U, X)}\times \dfrac{P(S=1\mid Y=0, A=0, U, X)}{P(S=1\mid Y=1, A=0, U, X)}\\\stackrel{A.\ref{assump:selection}}{=}& \exp(\beta_0)\times \xi(U, X)\times \xi(U, X)^{-1} = \exp(\beta_0).
\end{align*}
Therefore 
\begin{align*}
    \theta_{1c} &=  E\left[c(X)\dfrac{P(A=1\mid U, X,Y=1, S=1)}{P(A=1\mid U, X, Y=0, S=1)}P(Y=1\mid U, X, S=1)\,\mid \, S=1\right] \\
    &=  E\left[c(X)\dfrac{P(A=0\mid U, X,Y=1, S=1)}{P(A=0\mid U, X, Y=0, S=1)}P(Y=1\mid U, X, S=1)\, \mid \, S=1\right]\exp(\beta_0) \\
    &= \theta_{0c}\exp(\beta_0),
\end{align*}
which proves Lemma~\ref{lemma:estimand}.

\end{proof}

By Lemma~\ref{lemma:estimand}, it suffices to prove that $\theta_{ac}=E[c(X)I(A=a)Yq(A, Z, X)\mid S=1]$. Under Assumptions~\eqref{assump:nc} and \eqref{assump:completeness-a}, $q(A,Z,X)$ solves Equation~\eqref{eq:trt-bridge-u}. We have

\begin{align}
    & E[c(X)I(A=a)Yq(a, Z, X)\mid S=1]\nonumber\\
    =&  E[c(X)I(A=a)YE\left(q(a, Z, X)\mid A=a, U, X, Y=1, S=1\right)\mid S=1]\nonumber\\
    \stackrel{A.\ref{assump:nc}}{=}& E[c(X)I(A=a)YE\left(q(a, Z, X)\mid A=a, U, X, Y=0, S=1\right)Y\mid S=1]\nonumber\\
     \stackrel{A.\ref{assump:trt-bridge}}{=}& E[c(X)I(A=a)Y\dfrac{1}{P(A=a\mid U, X, Y=0, S=1)}\mid S=1]\nonumber\\
     =& E[c(X)E\{I(A=a)\mid U, X, Y=1, S=1\}\dfrac{1}{P(A=a\mid U, X, Y=0, S=1)}Y\mid S=1]\nonumber\\
     =& E[c(X)\dfrac{P(A=a\mid U, X, Y=1, S=1)}{P(A=a\mid U, X, Y=0, S=1)}Y\mid S=1]\nonumber\\
     =& E[c(X)\dfrac{P(A=a\mid U, X, Y=1, S=1)}{P(A=a\mid U, X, Y=0, S=1)}P(Y=1\mid  U, X, S=1)\mid S=1]\nonumber\\
     =& \theta_{ac}.\label{eq:pf-tac-PIPW}
\end{align}

\clearpage
\subsection{Proof of Theorem~\ref{thm:outcome-bridge}}

First, We show that Equation~\eqref{eq:h-odds-u} is equivalent to the moment equation

\begin{equation}\label{eq:outcome-bridge-u}
    E\{(1-Y)h(a, W, X)\mid A=a, U, X, S=1\} = E(Y\mid A=a, U, X, S=1).
\end{equation}

This is because
\begin{align*}
    & E\{(1-Y)h(a, W, X)\mid A=a, U, X, S=1\}\\
    =& E\{h(a, W, X)\mid A=a, U, X, Y=0, S=1\}P(Y=0\mid A=a, U, X, S=1)\\
    \stackrel{A.\ref{assump:outcome-bridge}}{=}&\dfrac{P(Y=1\mid A=a, U, X, S=1)}{P(Y=0\mid A=a, U, X, S=1)}P(Y=0\mid A=a, U, X, S=1)\\
    =& P(Y=1\mid A=a, U, X, S=1)\\
    =& E(Y\mid A=a, U, X, S=1).
\end{align*}

Next, we show that Assumption~\ref{assump:nc} and Equation~\eqref{eq:outcome-bridge-u} implies the moment equation
\begin{equation}\label{eq:outcome-bridge}
    E\{(1-Y)h(a, W, X)\mid A=a, Z, X, S=1\} = E(Y\mid A=a, Z, X, S=1).
\end{equation}

This is because
\begin{align*}
    &E\{(1-Y)h(A, W, X) - Y\mid A, Z, X, S=1\}\\
    = &E[E\{(1-Y)h(A, W, X) - Y\mid A, U, Z, X, S=1\}\mid A, Z, X, S=1]\\
    \stackrel{A.\ref{assump:nc}}{=}& E[E\{(1-Y)h(A, W, X) - Y\mid A, U, X, S=1\}\mid A, Z, X, S=1]\\
    \stackrel{Eq\eqref{eq:outcome-bridge-u}}{=}& 0.
\end{align*}
Finally, the left-hand side of Equation~\eqref{eq:outcome-bridge} equals
$$E\{h(a, W, X)\mid A=a, Z, X, Y=0, S=1\}P(Y=0\mid A=a, Z, X, S=1).$$

Rearranging the terms, we obtain Equation~\eqref{eq:h-odds}.

To prove the uniqueness of a solution to Equation~\eqref{eq:h-odds} under Assumption~\ref{assump:completeness-b}(b), we suppose there exists another function $h^*$ that solves Equation~\eqref{eq:h-odds}, that is,
$$E\{h^*(a,W,X)\mid A=a, Z, X,Y=0, S=1\} = \frac{P(Y=1\mid A=a,Z,X,S=1)}{P(Y=0\mid A=a,Z,X, S=1)}$$
Then
$$E\{h^*(a, W, X)-h(a, W, X)\mid A=a, Z, X, Y=0, S=1\} = 0.$$

By Assumption~\ref{assump:completeness-b}(b), $h^*(A, Z, X)=h(A, Z, X)$ almost surely.

Suppose there exists a function $h_0$ that solves Equation~\eqref{eq:h-odds-u}, then $h_0$ is also the solution to~\eqref{eq:h-odds} by the proof above. By the uniqueness of a solution to Equation~\eqref{eq:h-odds}, $h_0(A, Z, X)=h(A,Z,X)$ almost surely, which proves that $h$ is the unique solution to Equation \eqref{eq:h-odds-u}.

\clearpage
\subsection{Proof of Theorem~\ref{thm:POR}}
  Under Model~\ref{mod:logit-model} and Assumptions~\ref{assump:selection} and \ref{assump:nc}, the result in Lemma~\ref{lemma:estimand} holds. It suffices to prove that
\begin{equation}\label{eq:pf-POR}
    \theta_{ac}=E\{c(X)(1-Y)h(a, W, X)\mid S=1\}
\end{equation}
for $a=0,1$.

For $a=0, 1$, we have that 
\begin{align*}
    & E[c(X)(1-Y)h(a, W, X)\mid S=1]\nonumber\\
    =& E[c(X)(1-Y)E\{h(a, W, X)\mid U, X, Y=0, S=1\}\mid S=1]\nonumber\\
    \stackrel{A.\ref{assump:nc}}{=}& E[c(X)(1-Y)E\{h(a, W, X)\mid A=a, U, X, Y=0, S=1\}\mid S=1]\nonumber\\
    \stackrel{A.\ref{assump:outcome-bridge}}{=}& E[c(X)(1-Y)\dfrac{P(Y=1\mid A=a, U, X, S=1)}{P(Y=0\mid A=a, U, X, S=1)}\mid S=1]\nonumber\\
    =& E[c(X)E(1-Y\mid U, X, S=1)\dfrac{P(Y=1\mid A=a, U, X, S=1)}{P(Y=0\mid A=a, U, X, S=1)}\mid S=1]\nonumber\\
    =& E[c(X)P(Y=0\mid U, X, S=1)\dfrac{P(Y=1\mid A=a, U, X, S=1)}{P(Y=0\mid A=a, U, X, S=1)}\mid S=1]\nonumber\\
    =& E[c(X) \dfrac{P(A=a\mid U, X, Y=1, S=1)}{P(A=a\mid U, X, Y=0, S=1)}P(Y=1\mid U, X, S=1)\mid S=1]\nonumber\\
    =& \theta_{ac}.
\end{align*}

\clearpage

\clearpage
\subsection{Proof of Theorem~\ref{thm:PDR}}

To prove
 \begin{align}
\beta_0 &= \log\bigg(E\bigg\{c(X)\bigg[I(A=1)q(A, Z, X)\left\{Y - (1-Y)h(A, W, X)\right\}+\nonumber\\
&\qquad (1-Y)h(1, W, X)\bigg]\mid S=1\bigg\} /E\bigg\{c(X)\bigg[I(A=1)q(A, Z, X)\{Y - \nonumber\\
&\qquad (1-Y)h(A, W, X)\}+(1-Y)h(0, W, X)\bigg]\mid S=1\bigg\}\bigg),\nonumber
\end{align}
by Lemma~\ref{lemma:estimand} it suffices to prove

\begin{align}
    \theta_{ac}&=E\big\{c(X)\big[(1-Y)h(a, W, X)-\nonumber\\&\qquad I(A=a)q(a, Z, X)\big\{(1-Y)h(A, W, X)-Y\big\}\big]\mid S=1\big\}.\label{eq:pf-tac}
\end{align}
for $a=0,1$.
\begin{enumerate}[(a)]
    \item If Assumptions~ \ref{assump:nc} and Equation~\eqref{eq:trt-bridge-u} hold, then $q$ solves Equation~\eqref{eq:trt-bridge} and the conclusion in Theorem~\ref{thm:PIPW} holds. 
    
    The right-hand side of Equation~\eqref{eq:pf-tac} equals
    \begin{align*}
  & E\bigg\{c(X)\big[(1-Y)h(a, W, X)-I(A=a)(1-Y)q(a, Z, X)h(A, W, X)+\\&\qquad I(A=a)q(a, Z, X)Y\big]\mid S=1\bigg\}\\
    \stackrel{Eq\eqref{eq:pf-tac-PIPW}}{=}& E\bigg\{c(X)(1-Y)h(a, W, X)\big[1-I(A=a)q(a, Z, X)\big]\mid S=1\bigg\}+\theta_{ac}\\
    =&  E\bigg\{c(X)(1-Y)h(a, W, X)E\big[1-\\&\qquad I(A=a)q(a, Z, X)\mid A, W, X, Y=0, S=1\big]\mid S=1\bigg\}+ \theta_{ac}\\
    \stackrel{Eq\eqref{eq:trt-bridge}}{=}&E\bigg\{c(X)(1-Y)h(a, W, X)\big[1-\dfrac{I(A=a)}{P(A=a\mid W, X, Y=0, S=1)}\big]\mid S=1\bigg\}+\theta_{ac}\\
    =&E\bigg\{c(X)(1-Y)\big[h(a, W, X)-\dfrac{I(A=a)h(a, W, X)}{P(A=a\mid W, X, Y=0, S=1)}\big]\mid S=1\bigg\}+\theta_{ac}\\
    =& E\bigg\{c(X)(1-Y)E\big[h(a, W, X)-\\&\qquad\dfrac{I(A=a)h(a, W, X)}{P(A=a\mid W, X, Y=0, S=1)}\mid W, X, Y=0, S=1\big]\mid S=1\bigg\}+\theta_{ac}\\
    =& E\bigg\{c(X)(1-Y)\big[h(a, W, X)-h(a, W, X)\big]\mid S=1\bigg\}+\theta_{ac}\\
    =& \theta_{ac}.
\end{align*}
    
\item If Assumptions~\ref{assump:nc} and Equation~\eqref{eq:h-odds-u} hold, then $h$ solves Equation~\eqref{eq:outcome-bridge} and the conclusion in Theorem~\ref{thm:POR} holds. The right-hand side of Equation~\eqref{eq:pf-tac} equals

\begin{align*}
    &E\bigg\{c(X)\big[(1-Y)h(a, W, X)-I(A=a)(1-Y)q(a, Z, X)h(a, W, X)+\\&\qquad I(A=a)q(a, Z, X)Y\big]\mid S=1\bigg\}\\
    =& E\bigg\{c(X)(1-Y)h(a, W, X)\mid S=1\bigg\}- \\&\qquad E\bigg\{c(X)I(A=a)q(a, Z, X)[(1-Y)h(A, W, X)-Y]\mid S=1\bigg\}\\
    \stackrel{Eq\eqref{eq:pf-POR}}{=}& \theta_{ac} - E\bigg\{c(X)I(A=a)q(a, Z, X)\times\\&\qquad E[(1-Y)h(A, W, X)-Y\mid A=a, Z, X, S=1]\mid \\&\qquad S=1\bigg\}\\
    \stackrel{Eq\eqref{eq:outcome-bridge}}{=}& \theta_{ac}-0\\=&\theta_{ac}.
\end{align*}
    
\end{enumerate}

\clearpage
\subsection{Proof of Theorem~\ref{thm:q-h-ident}}

To prove Equation~\eqref{eq:q-identification}, we have
\begin{align*}
    & E\{(1-Y)[\kappa_1(A, W, X)q(A, Z, X)-\kappa_1(1, W, X) - \kappa_1(0, W, X)]\mid S=1\}\\
    =&  E\{(1-Y)[\kappa_1(A, W, X)E[q(A, Z, X)\mid A, W, X, Y=0, S=1]-\kappa_1(1, W, X) - \\&\qquad \kappa_1(0, W, X)] \mid S=1\}\\
    \stackrel{Eq\eqref{eq:trt-bridge}}{=}& E\{(1-Y)[\dfrac{\kappa_1(A, W, X)}{P(A\mid W, X, Y=0, S=1)}-\kappa_1(1, W, X) - \kappa_1(0, W, X)]\mid  S=1\}\\
    =& E\{(1-Y)[E\left(\dfrac{\kappa_1(A, W, X)}{P(A\mid W, X, Y=0, S=1)}\mid W, X, Y=0, S=1\right)-\kappa_1(1, W, X) - \\&\qquad\kappa_1(0, W, X)]\mid  S=1\}\\
    =& E\{(1-Y)[\kappa_1(1, W, X) + \kappa_1(0, W, X)-\kappa_1(1, W, X) - \kappa_1(0, W, X)]\mid S=1\}\\
    =& 0.
\end{align*}

Equation~\eqref{eq:h-identification} directly follows Equation~\eqref{eq:outcome-bridge}.

\clearpage
\subsection{Proof of Theorem~\ref{thm:doubly-robust}}
Under Model~\ref{mod:logit-model} and Assumptions~\ref{assump:selection} and \ref{assump:nc}, Lemma~\ref{lemma:estimand} holds.
\begin{enumerate}[(a)]
\item If $q(A, Z, X;\tau)$ is correctly specified and Assumption~\ref{assump:completeness-a}(b) holds, then Equation~\eqref{eq:q-identification} has a unique solution that is $\tau_0$ such that $q(A, Z, X;\tau_0)=q(A, Z, X)$, and $q(A, Z, X;\hat\tau)$ is consistent for $q(A, Z, X)$. By Law of Large Numbers and Continuous Mapping Theorem, we have
\begin{align*}
    \hat\beta_{c, \text{PIPW}}&\overset{p}{\rightarrow}\log\left(\dfrac{E\left\{c(X)I(A=1)Yq(A, Z, X)\mid S=1\right\}}{E\left\{c(X)I(A=0)Yq(A, Z, X)\mid S=1\right\}}\right)\\
    &\stackrel{Eq\eqref{eq:pf-tac-PIPW}}{=}\log(\theta_{1c}/\theta_{0c})\\
    &= \beta_0.
\end{align*}

\item If $h(A, W, X;\psi)$ is correctly specified and Assumption~\ref{assump:completeness-b}(b) holds, then Equation~\eqref{eq:h-identification} has a unique solution that is $\psi_0$ such that $h(A, W, X;\psi_0)=h(A, W,X)$, and $h(A, W, X;\hat\psi)$ is consistent for $h(A, W, X)$. By   Law of Large Numbers and Continuous Mapping Theorem, we have
\begin{align*}
    \hat\beta_{c, \text{POR}}&\overset{p}{\rightarrow}log\left(\dfrac{E\left\{c(X)(1-Y)h(1, W, X)\mid S=1\right\}}{E\left\{c(X)(1-Y)h(0, W, X)\mid S=1\right\}}\right)\\
    &\stackrel{Eq\eqref{eq:pf-POR}}{=}\log(\theta_{1c}/\theta_{0c})\\
    &= \beta_0.
\end{align*}

\item Write $q^*(A, Z, X)$ and $h^*(A, W, X)$ as the probability limit of $q(A, Z, X;\hat\tau)$ and $h(A, W, X;\hat\psi)$, respectively. The stated result follows Law of Large Numbers, Continuous Mapping Theorem, and the proof of Theorem~\ref{thm:PDR}.

\end{enumerate}

\clearpage
\section{Approximate equivalence of RR and OR for a rare outcome}\label{supp:or-and-rr}
Under Model~\ref{mod:logit-model}, the conditional odds ratio is
$$\text{OR}(U,X)=\exp(\beta_0)=\dfrac{P(Y=1\mid A=1, U, X)P(Y=0\mid A=0, U, X)}{P(Y=1\mid A=0, U, X)P(Y=1\mid A=1, U, X)}.$$

Consider the conditional risk ratio
$$\text{RR}(U,X):=\dfrac{P(Y=1\mid A=1, U, X)}{P(Y=1\mid A=0, U, X)} =\text{OR}(U,X)\times \dfrac{P(Y=0\mid A=1, U, X)}{P(Y=0\mid A=0, U, X)}.$$

Under a rare infection assumption that 
$$0\leq P(Y=1\mid A=a, U,X)\leq \delta$$
almost surely for $a=0,1$ and a small $\delta >0$, OR and RR satisfies
\begin{equation}\label{eq:or-rr}1-\delta\leq\dfrac{\text{RR}(U,X)}{\text{OR}(U,X)}=\dfrac{P(Y=0\mid A=1, U, X)}{P(Y=0\mid A=0, U, X)}\leq \dfrac{1}{1-\delta}.\end{equation}

If $\delta\approx 0$, then $\text{OR}\approx\text{RR}$.

Furthermore, denote $Y(a)$ as the potential outcome had the individual received the treatment $A=a$. We make the following standard identifiability assumptions~\citep{hernan2020causal}:
\begin{enumerate}[(a)]
    \item (Consistency) $Y(a)=Y$ if $A=a$;
    \item (Confounded ignorability) $Y(a)\indep A\mid U, X$ for $a=0,1$;
    \item (Positivity) $0<P(A=1\mid U, X)<1$ almost surely.
\end{enumerate}

Under these assumptions, the marginal causal risk ratio is
\begin{align*}
    \text{cRR} &:=\dfrac{P\{Y(1)=1\}}{P\{Y(0)=1\}} \\
    &= \dfrac{\int\int P\{Y(1)=1\mid U=u, X=x\}f(u, x)dudx}{\int\int P\{Y(0)=1\mid U=u, X=x\}f(u, x)dudx}\\
    &\stackrel{(b)}{=} \dfrac{\int\int P\{Y(1)=1\mid A=1,U=u, X=x\}f(u, x)dudx}{\int\int P\{Y(0)=1\mid A=0, U=u, X=x\}f(u, x)dudx}\\
    &\stackrel{(a)}{=} \dfrac{\int\int P\{Y=1\mid A=1,U=u, X=x\}f(u, x)dudx}{\int\int P\{Y=1\mid A=0, U=u, X=x\}f(u, x)dudx}\\
    &= \dfrac{\int\int \text{RR}(U, X)P\{Y=1\mid A=0,U=u, X=x\}f(u, x)dudx}{\int\int P\{Y=1\mid A=0, U=u, X=x\}f(u, x)dudx}\\
    &\stackrel{Eq\eqref{eq:or-rr}}{\leq} \dfrac{1}{1-\delta}\exp(\beta_0)
\end{align*}
and
$$  \text{cRR}\geq (1-\delta)\exp(\beta_0).$$

If $\delta\approx 0$, then $\text{cRR}\approx \exp(\beta_0)$.

\clearpage
\section{Identification of marginal causal risk ratio in a test-negative design with diseased controls}\label{supp:tnd}

Similar to \citet{schnitzer2022estimands}, we denote $D$ as a categorical variable indicating a subject's outcome, where $D=2$ indicates $Y=1$, i.e. the subject is positive for the infection of interest; $D=1$ indicates that the subject has $Y=0$ but is positive for another outcome, referred to as the ``control infection", that is known a priori to not be affected by $A$; and $D=0$ indicates that neither outcome is positive. To formalize the assumption of no treatment effect on the control disease, we assume 
\begin{equation}\label{eq:nco}
    P(D=1\mid A=1, U, X)=P(D=1\mid A=0, U, X).
\end{equation}
As sample selection is only restricted to subjects with infection, we have 
\begin{equation}\label{eq:sel-s1}
P(D=0\mid S=1, U, X, A)=0.
\end{equation}
Rather than Assumption~\ref{assump:selection}, we make the stronger assumption that 
\begin{equation}\label{eq:sel-s2}
A\indep S\mid D, U, X\quad\text{and}\quad A\indep S\mid Y, U, X.
\end{equation} That is, a subject's treatment status has no impact on the sample selection. This assumption holds in a test-negative study if, given a subject's disease status and other traits included in $(U,X)$, their decision to seek care does not depend on their treatment status.

Under Model~\ref{mod:logit-model} and above assumptions, we have
\begin{align*}
    \exp(\beta_0) &= \dfrac{P(A=1\mid Y=1, U, X)P(A=0\mid Y=0, U, X)}{P(A=1\mid Y=0, U, X)P(A=0\mid Y=1, U, X)}\\
     &\stackrel{Eq\eqref{eq:sel-s2}}{=} \dfrac{P(A=1\mid Y=1, U, X, S=1)P(A=0\mid Y=0, U, X, S=1)}{P(A=1\mid Y=0, U, X, S=1)P(A=0\mid Y=1, U, X, S=1)}\\
    &= \dfrac{P(A=1\mid D=2, U, X, S=1)P(A=0\mid D\neq 2, U, X, S=1)}{P(A=1\mid D\neq 2, U, X, S=1)P(A=0\mid D=2, U, X, S=1)}\\
    &\stackrel{Eq\eqref{eq:sel-s1}}{=} \dfrac{P(A=1\mid D=2, U, X, S=1)P(A=0\mid D=1, U, X, S=1)}{P(A=1\mid D=1, U, X, S=1)P(A=0\mid D=2, U, X, S=1)}\\
    &\stackrel{Eq\eqref{eq:sel-s2}}{=}  \dfrac{P(A=1\mid D=2, U, X)P(A=0\mid D=1, U, X)}{P(A=1\mid D=1, U, X)P(A=0\mid D=2, U, X)}\\
    &= \dfrac{P(D=2\mid A=1, U, X)P(D=1\mid A=0, U, X)}{P(D=1\mid A=1, U, X)P(D=2\mid A=0, U, X)}\\
    &\stackrel{Eq\eqref{eq:nco}}{=} \dfrac{P(D=2\mid A=1, U, X, )}{P(D=2\mid A=0, U, X)}\\
    &= \dfrac{P(Y=1\mid A=1, U, X )}{P(Y=1\mid A=0, U, X)}\\
    &\equiv \text{RR}(U,X)
\end{align*}

Under the standard identifiability assumptions in Section~\ref{supp:or-and-rr}, the marginal causal risk ratio is
\begin{align*}
    \text{cRR} &= \dfrac{\int\int \text{RR}(U, X)P\{Y=1\mid A=0,U=u, X=x\}f(u, x)dudx}{\int\int P\{Y=1\mid A=0, U=u, X=x\}f(u, x)dudx}\\
    &= \exp(\beta_0).
\end{align*}
In the above discussion, we do not invoke the rare outcome assumption.

\clearpage

\section{Regularity conditions for the existence of a solution to Equations~\eqref{eq:trt-bridge-u} and \eqref{eq:outcome-bridge-u}}\label{supp:bridge-existence}

The results in this section directly adapted from Appendix B of \citet{cui2020semiparametric} and Section B of the Supplementary Materials of \citet{li2022double}.

Let $L^2\{F(t)\}$ denote the Hilbert space of all square-integrable functions of $t$ with respect to distribution function $F(t)$, equiped with inner product $\langle g_1, g_2\rangle=\int g_1(t)g_2(t)dF(t)$. Let $T_{a,u,x}:L^2\{F(z|a,x, Y=0, S=1)\}\rightarrow L^2\{F(u|a, x, Y=0, S=1)\}$ denote the conditional expectation operator $T_{a,u,x}h=E[h(a, Z, x)|A=a, U=u, X=x, Y=0, S=1]$ and let $(\lambda_{a,u,x,n},\varphi_{a,u,x,n},\phi_{a,u,x,n})$ denote a singular value decomposition of $T_{a,u,x}$. Similarly, let $T'_{a,u,x}:L^2\{F(w|a,x, Y=0, S=1)\}\rightarrow L^2\{F(u|a, x, Y=0, S=1)\}$ denote the conditional expectation operator $T'_{a,u,x}g=E[g(a, W, x)|A=a, U=u, X=x, Y=0, S=1]$ and let $(\lambda'_{a,u,x,n},\varphi'_{a,u,x,n},\phi'_{a,u,x,n})$ denote a singular value decomposition of $T'_{a,u,x}$. Consider the following regularity conditions:
\begin{enumerate}[(1)]
    \item $\int\int f(z|a, u, x, Y=0, S=1)f(u|a,z,x, Y=0, S=1)dzdu<\infty$;
    \item $\int P^{-2}(A=a|U=u,X=x, Y=0, S=1)f(u|a,x, Y=0, S=1)du<\infty$;
    \item $\sum_{n=1}^\infty \lambda_{a,u,x,n}^{-2}|\langle P^{-1}(A=a|U=u,X=x, Y=0, S=1),\phi_{a,u,x,n}\rangle|^2<\infty$;
    \item $\sum_{n=1}^\infty (\lambda_{a,u,x,n}')^{-2}|\langle P^{-1}(A=a|U=u,X=x, Y=0, S=1),\phi_{a,u,x,n}'\rangle|^2<\infty$.
\end{enumerate}

The solution to Equation~\eqref{eq:trt-bridge-u} exists if Assumption~\ref{assump:completeness-a}(a) and regularity conditions (1), (2) and (3) hold. The solution to Equation~\eqref{eq:h-odds-u} exists if Assumption~\ref{assump:completeness-b}(a) and regularity conditions (1), (2) and (4) hold.
\clearpage
\section{Summary of our assumptions and results and those in \citet{li2022double}}\label{supp:assump}

\begin{table}[!h]
    \centering
    \caption{Summary of models, assumptions and identification results, with comparison to those in \citet{li2022double} for identifying conditional risk ratio}
    \resizebox{0.98\textwidth}{!}{
    \begin{tabular}{l|l|l|l}
    \hhline{====}
    \multicolumn{2}{c|}{\multirow{2}{*}{}}&\citet{li2022double} Sections 2.2-2.4 & PIPW identification  \\
     \multicolumn{2}{c|}{}& &(\citet{li2022double} Section 2.5)\\
    \hline
    \multicolumn{2}{l|}{Model} & $\log\left(\dfrac{P(Y=1\mid A=1, U, X)}{P(Y=1\mid A=0, U, X)}\right)=\beta_0$ & Model~\ref{mod:logit-model}\\
    \hline 
    \multicolumn{2}{l|}{Interpretation of $\beta_0$}& conditional risk ratio & conditional odds ratio \\
    \hline
      \multicolumn{2}{l|}{Selection mechanism} & $A\indep S\mid Y, U, X$ & Assumption~\ref{assump:selection}\\
    \hline
   \multicolumn{2}{l|}{\multirow{4}{*}{Proximal independence}} &$Z\indep (Y,S)|A, U, X$ & Assumption~\ref{assump:nc}\\
   \multicolumn{2}{l|}{} &  $W\indep A\mid U, X$ & \\
    \multicolumn{2}{l|}{} & $W \indep Z\mid A, U, X, Y$ & \\
    \multicolumn{2}{l|}{} & $S\indep Z\mid A, U, X, W, Y$ & \\
    \hline 
    \multicolumn{2}{l|}{$q$ definition} & $\begin{array}{c}E\{q(a, Z, X\mid A=a, U, X)\}\\\quad =1/P(A=a\mid U, X)\end{array}$ & Equation~\eqref{eq:trt-bridge-u}\\
    \hline 
    \multicolumn{2}{l|}{$h$ definition} & / & / \\
    \hline
     \multicolumn{2}{l|}{Completeness} & Assumption~\ref{assump:completeness-a}(b) & Assumption~\ref{assump:completeness-a}(b)\\
     \hline
   \multirow{3}{*}{Identification} & $q$ & Equation~\eqref{eq:trt-bridge} & Equation~\eqref{eq:trt-bridge}\\
   \cline{2-4}
   & $h$ & / & / \\
   \cline{2-4}
   & $\beta_0$ & Equation~\eqref{eq:beta-PIPW}& Equation~\eqref{eq:beta-PIPW}\\
   \hline
   \multicolumn{2}{l|}{Note} & (1) Identification of $q$ and $\beta_0$ requires a & Under standard identifiability   \\
   \multicolumn{2}{l|}{} & rare outcome condition; & assumptions in Section~\ref{supp:or-and-rr}, if either  \\
   \multicolumn{2}{l|}{} & (2) Under standard identifiability   & the outcome is rare or under the \\
   \multicolumn{2}{l|}{} & assumptions in Section~\ref{supp:or-and-rr}, $\beta_0$ & setting of Section~\ref{supp:tnd}, $\beta_0$ identifies the \\
   \multicolumn{2}{l|}{} & identifies the marginal causal risk ratio. & marginal causal risk ratio.\\
   \multicolumn{2}{l|}{} & (3) The proximal independence conditions  & \\
   \multicolumn{2}{l|}{} & can be replaced by Assumption~\ref{assump:nc}&\\
   \hhline{====}
    \multicolumn{2}{c|}{}& POR identification & PDR closed-form expression\\
    \hline
    \multicolumn{2}{l|}{Model} & Model~\ref{mod:logit-model} & Model~\ref{mod:logit-model}\\
    \hline 
    \multicolumn{2}{l|}{Interpretation of $\beta_0$}& conditional odds ratio & conditional odds ratio \\
    \hline
      \multicolumn{2}{l|}{Selection mechanism} & Assumption~\ref{assump:selection} & Assumption~\ref{assump:selection}\\
    \hline
   \multicolumn{2}{l|}{Proximal independence} &Assumption~\ref{assump:nc} & Assumption~\ref{assump:nc}\\
    \hline 
    \multicolumn{2}{l|}{$q$ definition} & / & Equation~\eqref{eq:trt-bridge-u}\\
    \hline 
    \multicolumn{2}{l|}{$h$ definition} & Equation~\eqref{eq:h-odds-u} & Equation~\eqref{eq:h-odds-u} \\
    \hline
     \multicolumn{2}{l|}{Completeness} & Assumption~\ref{assump:completeness-b}(b) & Assumptions~\ref{assump:completeness-a}(b), \ref{assump:completeness-b}(b)\\
     \hline
   \multirow{3}{*}{Identification} & $q$ & / & Equation~\eqref{eq:trt-bridge} + Assumption~\ref{assump:completeness-a}(b)\\
   \cline{2-4}
   & $h$ & Equation~\eqref{eq:h-odds} & Equation~\eqref{eq:h-odds} + Assumption~\ref{assump:completeness-b}(b) \\
   \cline{2-4}
   & $\beta_0$ & Equation~\eqref{eq:beta-POR}& Equation~\eqref{eq:beta-PDR}\\
   \hline
   \multicolumn{2}{l|}{Note} &  Under standard identifiability  a & (1) Under standard identifiability   \\
   \multicolumn{2}{l|}{} & assumptions in Section~\ref{supp:or-and-rr}, if either & assumptions in Section~\ref{supp:or-and-rr}, if either  \\
   \multicolumn{2}{l|}{} & the outcome is rare or under the & the outcome is rare or under the \\
   \multicolumn{2}{l|}{} & setting of Section~\ref{supp:tnd}, $\beta_0$ identifies the $\beta_0$ & setting of Section~\ref{supp:tnd}, $\beta_0$ identifies the \\
   \multicolumn{2}{l|}{} &  marginal causal risk ratio. & marginal causal risk ratio.\\
   \multicolumn{2}{l|}{}& & (2) Equation~\eqref{eq:beta-POR} holds if either $q$ is a \\
   \multicolumn{2}{l|}{}& & solution to Equation~\eqref{eq:trt-bridge-u} or $h$ is a \\
   \multicolumn{2}{l|}{}& & solution to Equation~\eqref{eq:h-odds-u}.\\
   \hline
    \end{tabular}
    }
    \label{tab:data_app}
\end{table}
\clearpage

\section{Detailed setting and derivation of $q(A, Z, X)$ and $h(A, W, X)$ in Example~\ref{eg:cont}}\label{supp:bridge-example-cont}

If $Z, W$ are both multivariate Gaussian random variables satisfying
\begin{align*}
    Z\mid A, U, X\sim N(\mu_{0Z}+\mu_{AZ}A+\mu_{UZ}^TU + \mu_{XZ}^TX,\Sigma_{Z}),\\
    W\mid A, U, X, Y\sim N(\mu_{0W} + \mu_{UW}^TU + \mu_{XW}^TX + \mu_{YW}Y,\Sigma_W).
\end{align*}
Assume that the treatment $A$ and outcome $Y$ both follow a logistic regression model given $(U, X)$:
\begin{align*}
    P(A=1\mid U, X) &= [1+\exp\{-(\mu_{0A}+\mu_{UA}^TU + \mu_{XA}^TX)\}]^{-1},\\
    P(Y=1\mid A, U, X) &= \exp(\mu_{0Y}+\beta_0A + \mu_{UY}^TU + \mu_{XY}^TX).
\end{align*}
Further assume the selection indicator $S$ follows a log-linear model of the form:
$$P(S=1\mid A, U, X, Y)=\exp(\mu_{AS}A + \mu_{YS}Y +\xi(U, X)).$$

It is straightforward to verify that the data distribution satisfies Model~\ref{mod:logit-model} and Assumptions~\ref{assump:selection}-\ref{assump:nc}. Then a suitable model for $h$ is
\begin{equation}
    h(a, w, x;\psi) = \exp(\psi_0 + \psi_{A}a + \psi_W^Tw + \psi_X^Tx).
\end{equation}
While a  closed-form parametric model for $q$ does not exist in this example, if the outcome $Y$ is rare in the target population in the sense that 
$$P(Y=1\mid A=a, U=u, X=x)\approx 0\qquad\text{almost surely},$$
then the parametric model below may serve as an approximation
\begin{equation}
    q(a, z, x;\tau) = 1+\exp\{(1-2A)(\tau_0+ \tau_{A}a + \tau_Z^Tz+\tau_{X}^Tx)\}.
\end{equation}

We first derive the outcome confounding bridge function $h$ that satisfies Equation~\eqref{eq:h-odds-u}:
\begin{align*}
    &E\{h(a,W,X)\mid A=a, U, X,Y=0, S=1\} =\\&\qquad P(Y=1\mid A=a, U,X,S=1)/P(Y=0\mid A=a, U,X, S=1).
\end{align*}

We have 
\begin{align*}
   & P(Y=1\mid A=a, Z, X)P(S=1\mid A=a, Y=1, Z, X)=\\
    &\qquad \dfrac{\exp(\mu_{0Y}+\beta_0a + \mu_{UY}U + \mu_{XY}^TX)}{1+\exp(\mu_{0Y}+\beta_0a + \mu_{UY}U + \mu_{XY}^TX)} \exp(\mu_{0S}+\mu_{AS}a+\mu_{YS} +\xi(U, X))\\
    & P(Y=0\mid A=a, U, X)P(S=1\mid A=a, Y=0, U=1, U, X)=\\
    &\qquad \dfrac{1}{1+\exp(\mu_{0Y}+\beta_0a + \mu_{UY}U + \mu_{XY}^TX)} \exp(\mu_{0S}+\mu_{AS}a+\xi(U, X)).
\end{align*}
Therefore
\begin{align*}
    &P(Y=1\mid A=a, U, X, S=1)\\
    =& P(Y=1\mid A=a, U, X)P(S=1\mid A=a, Y=1, U, X)/\{P(Y=1\mid A=a, U, X)\times \\&\qquad  P(S=1\mid A=a, Y=1, U, X) +P(Y=0\mid A=a, U, X)P(S=1\mid A=a, Y=0, U, X)\}\\
    =& [1+\exp\{-(\mu_{0Y}+\mu_{YS}+\beta_0a + \mu_{UY}U + \mu_{XY}^TX)\}]^{-1}.
\end{align*}

and
\begin{align*}
\dfrac{P(Y=1\mid A=a, U, X, S=1)}{P(Y=0\mid A=a, U, X, S=1)}=\exp(\mu_{0Y}+\mu_{YS}+\beta_0a + \mu_{UY}U + \mu_{XY}^TX).
\end{align*}

We further have $$W\mid A=a, U, X, Y=0, S=1 \sim N(\mu_{0W} + \mu_{UW}U + \mu_{XW}^TX).$$

Now assume that $h(a, w, x) = \exp(\psi_0 + \psi_Aa + \psi_Ww + \psi_X^Tx)$, then
\begin{align*}
    & E\{h(A, W, X)\mid A, U, X, Y=0, X=1\}\\
    =& \exp(\psi_0 + \psi_AA + \psi_X^TX)E\{\exp(\psi_WW)\mid A, U, X, Y=0, S=1\}\\
    =& \exp(\psi_0 + \psi_AA + \psi_X^TX) \exp(\psi_W^T(\mu_{0W} + \mu_{UW}^TU + \mu_{XW}^TX) + \frac{1}{2}\psi_W^T\Sigma_W\psi_W)\\
    =& \exp\{(\psi_0 + \psi_W^T\mu_{0W} + \frac{1}{2}\psi_W^T\Sigma_W\psi_W)+\psi_AA + \psi_W^T\mu_{UW}^TU + (\psi_X^T+\psi_W^T\mu_{XW}^T)X\}
\end{align*}

Hence $\psi$ satisfies
\begin{align*}
    &\psi_A = \beta_0,
    \qquad \psi_W = \mu_{UW}^{+}\mu_{UY}, \qquad \psi_X = \mu_{XY} - \mu_{XW}\psi_W = \mu_{XY} - \mu_{UY}\mu_{UW}^{+}\mu_{XW},\\
    &\psi_0 = \mu_{0Y} + \mu_{YS} - \psi_W^T\mu_{0W} - \dfrac{1}{2}\psi_W^T\Sigma_W\psi_W \\
    &\qquad =\mu_{0Y} + \mu_{YS} - \mu_{UW}^+\mu_{0W}\mu_{UY} - \dfrac{1}{2}\mu_{UY}^T(\mu_{UW}^+)^T\Sigma_W\mu_{UW}^+\mu_{UY}.
\end{align*}

Here $A^{+}$ denotes the generalised inverse of the matrix $A$. 

To derive $q(A, Z, X)$, we have
\begin{align*}
    & P(A=1\mid U, X)P(Y=0\mid A=1, U, X) P(S=1\mid A=1, Y=0, U, X)\\
    =& \dfrac{\exp(\mu_{0A}+\mu_{UA}^TU + \mu_{XA}^TX)}{1 + \exp(\mu_{0A}+\mu_{UA}^TU + \mu_{XA}^TX)}\dfrac{1}{1+\exp(\mu_{0Y}+\beta_0 + \mu_{UY}^TU + \mu_{XY}^TX)}\times\\&\qquad \exp\{\mu_{AS}+\xi(U, X)\},\\
    \text{and}\qquad & P(A=0\mid U, X)P(Y=0\mid A=0, U, X) P(S=1\mid A=0, Y=0, U, X)\\
    =& \dfrac{1}{1 + \exp(\mu_{0A}+\mu_{UA}^TU + \mu_{XA}^TX)}\dfrac{1}{1+\exp(\mu_{0Y}+ \mu_{UY}^TU + \mu_{XY}^TX)}\times\\&\qquad \exp\{\xi(U, X)\}.\\
\end{align*}
Therefore
\begin{align*}
    &P(A=1\mid U, X, Y=0, S=1)\\ 
    =& P(A=1\mid U, X)P(Y=0\mid A=1, U, X)P(S=1\mid Y=0, S=1, U, X)/\\ &\qquad\{P(A=1\mid U, X)P(Y=0\mid A=1, U, X)P(S=1\mid Y=0, S=1, U, X)+\\&\qquad  P(A=0\mid U, X)P(Y=0\mid A=0, U, X) P(S=1\mid A=0, Y=0, U, X)\}\\
    =& \dfrac{\exp(\mu_{0A} + \mu_{AS} + \mu_{UA}^TU + \mu_{XA}^TX)}{\exp(\mu_{0A} + \mu_{AS} + \mu_{UA}^TU + \mu_{XA}^TX)+\dfrac{1 + \exp(\mu_{0Y}+\beta_0 +\mu_{UY}^TU + \mu_{XY}^TX)}{1 + \exp(\mu_{0Y} +\mu_{UY}^TU + \mu_{XY}^TX)}}
\end{align*}

In the case where $P(Y=1\mid A, U, X)$ are small almost surely, both $\exp(\mu_{0Y}+\beta_0 +\mu_{UY}^TU + \mu_{XY}^TX)$ and $\exp(\mu_{0Y} +\mu_{UY}^TU + \mu_{XY}^TX)$ are necessarily small, then
$$\dfrac{1 + \exp(\mu_{0Y}+\beta_0 +\mu_{UY}^TU + \mu_{XY}^TX)}{1 + \exp(\mu_{0Y} +\mu_{UY}^TU + \mu_{XY}^TX)}\approx 1,$$
and therefore
$$P(A=1|U, X, Y=0, S=1)\approx [1 + \exp\{-(\mu_{0A} + \mu_{AS} + \mu_{UA}^TU + \mu_{XA}^TX)\}]^{-1}.$$

We need find a function $q(a, z, x)$ that satisfies Equation~\eqref{eq:trt-bridge-u}:

\begin{align*}
    E[q(1, Z, X)|A=1, U, X] = 1 + \exp\{-(\mu_{0A} + \mu_{AS} + \mu_{UA}^TU + \mu_{XA}^TX)\}
\end{align*}
and 
\begin{align*}
    E[q(0, Z, X)|A=0, U, X] = 1 + \exp(\mu_{0A} + \mu_{AS} + \mu_{UA}^TU + \mu_{XA}^TX).
\end{align*}

We consider $q(A, Z, X)=1 + \exp\{(1-2A)(\tau_0 + \tau_A A + \tau_Z^T Z + \tau_X^TX)\}$, then

\begin{align*}
    &E[q(1, Z, X)|A=1, U, X] \\=& 1 + \exp(-\tau_0 - \tau_A  - \tau_X^TX) E[\exp(-\tau_Z^TZ)|A=1, U, X]\\
    =& 1 + \exp(-\tau_0 - \tau_A  - \tau_X^TX)\exp\{-\tau_Z^T(\mu_{0Z} + \mu_{AZ} + \mu_{UZ}^TU + \mu_{XZ}^TX) + \frac{1}{2}\tau_Z^T\Sigma_Z\tau_Z\}\\
    =& 1+\exp[-\{(\tau_0 +\tau_A +  \tau_Z^T\mu_{0Z} + \tau_Z^T\mu_{AZ}-\frac{1}{2}\tau_Z^T\Sigma_Z\tau_Z) + \tau_Z^T\mu_{UZ}^TU + (\tau_X^T + \tau_Z^T\mu_{XZ}^T)X\}] 
\end{align*}
and
\begin{align*}
    & E[q(0, Z, X)|A=0, U, X]\\
    =& 1 + \exp(\tau_0 + \tau_X^TX)E[\exp(\tau_Z^TZ)\mid A=0, U, X]\\
    =& 1 + \exp(\tau_0 + \tau_X^TX)\exp\{\tau_Z^T(\mu_{0Z} + \mu_{UZ}^TU + \mu_{XZ}^TX) + \frac{1}{2}\tau_Z^T\Sigma_Z\tau_Z\}\\
    =& 1 + \exp\{(\tau_0 + \tau_Z^T\mu_{0Z} + \frac{1}{2}\tau_Z^T\Sigma_Z\tau_Z) + \tau_Z^T\mu_{UZ}^TU + (\tau_X^T + \tau_Z^T\mu_{XZ}^T)X\}
\end{align*}
Hence $\tau$ satisfies
\begin{align*}
    &\tau_Z = \mu_{UA}^+\mu_{UZ}, \qquad \tau_X = \mu_{XA} - \mu_{XZ}\tau_Z = \mu_{XA} - \mu_{XZ}\mu_{UA}^+\mu_{UZ},\\
    &\tau_A = \tau_Z^T\Sigma_Z\tau_Z,\\
    &\tau_0 = \mu_{0A} + \mu_{AS}-\tau_Z^T\mu_{0Z} - \dfrac{1}{2}\tau_Z^T\Sigma_Z\tau_Z.
\end{align*}

\clearpage
\section{Cross-fitting kernel learning estimator for nonparametric estimation of $\beta_0$.}\label{supp:RKHS}

Under Model~\ref{mod:logit-model} and Assumptions~\ref{assump:selection} and \ref{assump:nc}, by Lemma~\ref{lemma:estimand}, estimation of $\beta_0$ reduces to estimation of $\theta_{ac}$. We first derive an influence function of $\theta_{ac}$.

\begin{theorem}\label{thm:if-zeta}
An influence function for $\zeta_{ac}$ is  \begin{align}IF_{\zeta_{ac}}(O;q,h)&= \dfrac{S}{P(S=1)}(1-Y)c(X)h(a, W, X)-\nonumber\\&\qquad \dfrac{S}{P(S=1)}I(A=a)q(a, Z, X)c(X)\left[(1-Y)h(A, W, X)-Y\right]-\nonumber\\&\qquad\zeta_{ac},\nonumber\end{align}
\end{theorem}

\begin{proof}
We let $\mu_{ac}= E\left[Sc(X)I(A=a)Yq(A, Z, X)\mid S=1\right].$ 
Because $\mu_{ac}=\zeta_{ac}P(S=1)$, their influence functions $IF_{\mu_{ac}}$ and $IF_{\zeta_{ac}}$ satisfy
$$IF_{\theta_{ac}}(O;q,h)=IF_{\mu_{ac}}(O;q,h)/P(S=1).$$
We proceed to derive $IF_{\mu_{ac}}$.

Consider regular one-dimensional parametric submodel indexed by $t$, where $t=0$ indicates the truth, we have
\begin{align*}
    \dfrac{\partial\mu_{1ct}}{\partial t} &= E[c(X)\dot q(A, Z, X)AYS] + E[c(X)q(A, Z, X)AYSR(O)],
\end{align*}
where $R$ denotes the score function.

To calculate $E[c(X)\dot q(A, Z, X)AYS]$, we note that
\begin{align*}
    0 = E_t\left[S(1-Y)\left(q_t(A, Z, X) - \dfrac{1}{f_t(A\mid W, X, Y=0, S=1)}\right)n(A, W, X)\right].
\end{align*}
Taking derivatives w.r.t $t$ on both sides and rearranging the terms, we have
\begin{align}
    & E[S(1-Y)\dot q(A, Z, X)n(A, W, X)] \label{eq:q-dot}\\
    &\qquad= E\left\{S(1-Y)\left[\dfrac{1}{f(A\mid W, X, Y=0, S=1)}-q(A, Z, X)\right]n(A, X, W)R(O)\right\}\\
    &\qquad-E\left[S(1-Y)\dfrac{R(A\mid W, X, Y=0, S=1)}{f(A\mid W, X, Y=0, S=1)}n(A, W, X)\right]\label{eq:term3}
\end{align}

In Equation~\eqref{eq:q-dot}, taking $n(A, W, X)=Ah(A, W, X)c(X)$, we have
\begin{align*}
    &E[S(1-Y)c(X)\dot q(A, Z, X)Ah(A, W, X)]\\
    = & E\left\{ASc(X)\dot q(A, Z, X)E\left[(1-Y)h(A, W, X)\mid A, Z, X, S\right]\right\}\\
    \stackrel{\eqref{eq:h-identification}}{=}& E\left\{ASc(X)\dot q(A, Z, X)E\left[Y\mid A, Z, X, S=1\right]\right\}\\
    =& E[c(X)\dot q(A, Z, X)AYS]
\end{align*}

In Equation~\eqref{eq:term3}, by taking $n(A, W, X)=Ah(A, W, X)c(X)$, we have
\begin{align*}
&E\left[S(1-Y)c(X)\dfrac{R(A\mid W, X, Y=0, S=1)}{f(A\mid W, X, Y=0, S=1)}Ah(A, W, X)\right]\\
=& E\left[S(1-Y)c(X)\dfrac{R(O)}{f(A\mid W, X, Y=0, S=1)}Ah(A, W, X)\right]\\
&\qquad -E\left\{S(1-Y)c(X)R(W, X, Y, S)\dfrac{Ah(A, W, X)}{f(A\mid W, X, Y=0, S=1)}\right\}\\
=& E\left[S(1-Y)c(X)\dfrac{R(O)}{f(A\mid W, X, Y=0, S=1)}Ah(A, W, X)\right]\\
&\qquad -E\bigg\{S(1-Y)c(X)R(W, X, Y, S)E\bigg[\dfrac{Ah(A, W, X)}{f(A\mid W, X, Y=0, S=1)}\mid W, X, \\&\qquad Y=0, S=1\bigg]\bigg\}\\
=& E\left[S(1-Y)c(X)\dfrac{R(O)}{f(A\mid W, X, Y=0, S=1)}Ah(A, W, X)\right] -\\&\qquad E\left\{S(1-Y)c(X)R(W, X, Y, S)h(1, W, X)\right\}\\
=& E\left\{S(1-Y)c(X)\left[\dfrac{Ah(A, W, X)}{f(A\mid W, X, Y=0, S=1)}-h(1, W, X)\right]R(O)\right\}
\end{align*}

We therefore have
\begin{align*}
    &\quad E[c(X)\dot q(A, Z, X)AYS] \\&= E\left\{S(1-Y)c(X)\left[\dfrac{1}{f(A\mid W, X, Y=0, S=1)}-q(A, Z, X)\right]Ah(A, W, X)R(O)\right\}\\&\qquad-E\left\{S(1-Y)c(X)\left[\dfrac{Ah(A, W, X)}{f(A\mid W, X, Y=0, S=1)}-h(1, W, X)\right]R(O)\right\}\\
    &= E\left\{S(1-Y)c(X)\left[h(1, W, X)-Aq(A, Z, X)h(A, W, X)\right]R(O)\right\}
\end{align*}

Hence
\begin{align*}
    \dfrac{\partial\mu_{1ct}}{\partial t} &= E\bigg\{Sc(X)\bigg((1-Y)\bigg[h(1, W, X)-Aq(A, Z, X)h(A, W, X)\bigg]+\\&\qquad YAq(A, Z, X)\bigg)R(O)\bigg\}\\
    &= E\bigg\{Sc(X)\bigg((1-Y)\left[h(1, W, X)-Aq(A, Z, X)h(A, W, X)\right]+YAq(A, Z, X) -\\&\qquad  \mu_{1c}\bigg)R(O)\bigg\}
\end{align*}
and 
\begin{align*}IF_{\mu_{1c}}(\cdot)&= S(1-Y)c(X)h(1, W, X)-SAq(A, Z, X)c(X)\left[(1-Y)h(A, W, X)-Y\right]-\mu_{1c}.\end{align*}

Similarly,
\begin{align*}
IF_{\mu_{0c}}(\cdot)&= Sc(X)\bigg\{(1-Y)\bigg[h(0, W, X)-(1-A)q(A, Z, X)h(A, W, X)\bigg]+\\&\qquad Y(1-A)q(A, Z, X)\bigg\}\\
&= S(1-Y)c(X)h(0, W, X) - S(1-A)q(A, Z, X)c(X)\bigg[(1-Y)h(A, W, X)-\\&\qquad Y\bigg]- \mu_{0c}.\end{align*}
\end{proof}

The influence function $IF_{\zeta_{ac}}(O;q,h)$ can be written as

\begin{align}
IF_{\zeta_{ac}}(O;q_a,h_a)&=\dfrac{S}{P(S=1)}\big\{g_{a1}(A, Y, X)q_a(Z, X)h_a(W, X)+\nonumber\\&\qquad g_{a2}(Y, X)q_a(Z, X)+g_{a3}(Y, X)h_a(W, X)\big\}-\zeta_{ac},\label{eq:if2}
\end{align}
where \begin{align*}
    &g_{a1}(A, Y, X)=-c(X)I(A=a)(1-Y),\qquad g_{a2}(A, Y, X)=c(X)I(A=a)Y,\\ &g_{a3}(A, Y, X)=c(X)(1-Y), \qquad h_a(W, X) = h(a, W, X),\\
    \text{and}\qquad &   q_a(Z, X) = q(a, Z, X).
\end{align*}

This influence function belongs to the class of doubly-robust influence functions considered in \citet{ghassami2021minimax}, and thus their method directly applies. By Theorem 1 of \citet{rotnitzky2021characterization}, such an influence function also satisfies the so-call "mixed-bias" structure, and therefore $\sqrt n$-consistent doubly robust estimation is possible even when the nuisance functions are estimated nonparametrically. In this section, we describe a cross-fitting minimax kernel learning estimator for $\zeta_{ac}$ and stated the theoretic property of the resulting estimator of $\beta_{0}$.

We randomly partition the data evenly into $L$ subsamples. For each $l\in\{1,\dots,L\}$, let $I_l$ and $n_l$ be the indices and size of data in the $l$th sample, respectively. Let $m_l=n-n_l$ be the size of the sample excluding $I_l$. Write $q_a(Z, X)=q(a, Z, X)$ and $h_a(W, X)=h(A, W, X)$.  As proposed in \citet{ghassami2021minimax}, we estimate the nuisance functions $q_a$ and $h_a$ using data from $I_l^c$ by solving the following optimization problems:

\begin{align}
    \tilde q_{al}&=\arg\underset{q_a\in\mathcal Q_a}{\min}\,\underset{\dot h\in\mathcal H_a^*}{\max} \,\dfrac{1}{m_l}\sum_{i\in I_{l}^c} \bigg[ \dot h(W_i, X_i)\left\{q_a(Z_i, X_i)g_{a1}(A_i, Y_i, X_i)+g_3( Y_i, X_i)\right\} \nonumber\\ &\qquad -\dot h^2(W_i, X_i)\bigg] -\lambda_{\mathcal H_a^*}^{q_a}\lVert \dot h\rVert_{\mathcal H_a^*}^2 + \lambda_{\mathcal Q_a}^{q_a}\lVert q_a\rVert_{\mathcal Q_a}^2;\label{eq:opt-q}\\
    \tilde h_{al}&=\arg\underset{h_a\in\mathcal H_a}{\min}\,\underset{\dot q\in\mathcal Q_a^*}{\max} \,\dfrac{1}{m_l}\sum_{i\in I_{l}^c} \bigg[ \dot q(Z_i, X_i)\left\{h_a(W_i, X_i)g_{a1}(A_i, Y_i, X_i)+g_2(A_i, Y_i, X_i)\right\} \nonumber\\ &\qquad -\dot q^2(Z_i, X_i)\bigg] -\lambda_{\mathcal Q_a^*}^{h_a}\lVert \dot q\rVert_{\mathcal Q_a^*}^2 + \lambda_{\mathcal H_a}^{h_a}\lVert h_a\rVert_{\mathcal H_a}^2.\label{eq:opt-h}
\end{align}
where $\mathcal Q_a$, $\mathcal H_a$, $\mathcal Q_a^*$ and $\mathcal H_a^*$ are reproducing kernel Hilbert spaces (RKHS) generated by the kernel functions $K_{\mathcal Q_a}$, $K_{\mathcal H_a}$, $K_{\mathcal Q_a^*}$ an $K_{\mathcal H_a^*}$, equipped with norms $\lVert\cdot\rVert_{\mathcal Q_a}$, $\lVert\cdot\rVert_{\mathcal H_a}$, $\lVert\cdot\rVert_{\mathcal Q_a^*}$ and $\lVert\cdot\rVert_{\mathcal H_a^*}$, respectively, and $\lambda_{\mathcal Q_a}^{q_a}$, $\lambda_{\mathcal H_a}^{h_a}$, $\lambda_{\mathcal Q_a^*}^{h_a}$ and $\lambda_{\mathcal H_a^*}^{q_a}$ are non-negative regularizing parameters. The above optimization problems have closed form solutions, which we leave to to the end of this section. 

The objective function for $\tilde q_{al}$ can be viewed as a regularised version of empirical perturbation evaluated using data from $I_l^c$, by moving the influence function from the pair $(q_a,h_a)$ to $(q_a, h_a+\dot h)$, where the perturbation is defined as
\begin{align*}
    prt_{\theta_{ac}}(q_a,h_a;0,\dot h) &= IF_{\theta_{ac}}(A, Y, Z, W, X;q_a,h_a+\dot h) - IF_{\theta_{ac}}(A, Y, Z, W, X;q_a,h_a)\\
    &=\dfrac{S}{P(S=1)}\dot h(W, X)\left\{q_a(Z, X)g_{a1}(A, Y, X) - g_{a3}(Y, X)\right\}.
\end{align*}

The objective function contains two forms of regularization: the term $\dot h^2$ improves the robustness against model misspecification, stability and the convergence rates, while the Tikhonov-type penalties $\lambda_{\mathcal H_a}^{q_a}\lVert \dot h\rVert_{\mathcal H_a}^2$ and  $\lambda_{\mathcal Q_a}^{q_a}\lVert q_a\rVert_{\mathcal Q_a}^2$ penalise the complexity of the classes of nuisance functions considered. The objective function for $\tilde h_{al}$ can be interpreted similarly.

We estimate $\zeta_{ac}$ with the $l$th subsample with the PDR estimator
\begin{align}\label{eq:pdr-l}
    \tilde\zeta_{ac,l} &= \dfrac{1}{n_l}\sum_{i\in I_l} c(X_i)\big[(1-Y_i)\tilde h_{al}( W_i, X_i)- \nonumber\\&\qquad I(A_i=a)\tilde q_{al}(Z_i, X_i)\big\{(1-Y_i)\tilde h_{al}(W_i, X_i)-Y_i\big\}\big].
\end{align}
 The final estimator for $\zeta_{ac}$ is
$\tilde\zeta_{ac}=\sum_{l=1}^L \tilde\zeta_{ac,l}/L$,
which gives the cross-fitting minimax kernel learning estimator for $\beta_{0}$ as
$\tilde\beta_{0} = \log(\tilde\zeta_{1c} / \tilde\zeta_{0c}).$ We summarise the estimation of $\zeta_{ac}$ in Algorithm~\ref{alg:cf-kernel}.

\begin{algo}\label{alg:cf-kernel}
A cross-fitting minimax kernel learning algorithm to estimate $\beta_0$

\begin{tabbing}
   \qquad \enspace Split the data into $L$ even subsamples\\
   \qquad \enspace For $a=0$ to $a=1$\\
   \qquad \qquad For $l=1$ to $l=L$ \\
   \qquad \qquad \qquad  Solve $\tilde q_{al}$ and $\tilde h_{al}$ in Equations~\eqref{eq:opt-q} and \eqref{eq:opt-h}\\
  \qquad \qquad \qquad Obtain $\tilde\zeta_{ac,l}$ according to Equation~\eqref{eq:pdr-l} \\
\qquad \qquad $\tilde\zeta_{ac}\leftarrow \sum_{l=1}^L \tilde\zeta_{ac,l}/L$ \\
\qquad \enspace Output $\tilde\beta_{0} = \log(\tilde\zeta_{1c} / \tilde\zeta_{0c})$.
\end{tabbing}

\end{algo}

The estimator $\tilde\beta_{0c}$ is consistent and asymptotically normal, thus allowing valid statistical inference, as stated in Theorem~\ref{thm:asym} below.

\begin{theorem}\label{thm:asym}
Under Model~\ref{mod:logit-model} and Assumptions~\ref{assump:selection}-\ref{assump:completeness-b} and suitable regularity conditions listed in Section~\ref{supp:regularity-conditions}, the cross-fitting estimator $\tilde \beta_{0}$ in Algorithm~\ref{alg:cf-kernel} is consistent and asymptotically normal of $\beta_0$.
\end{theorem}

We prove Theorem~\ref{thm:asym} in Section~\ref{sec:proof-rkhs} of the Supplementary Materials.

Implementation of this approach is challenging, as the performance of the method proves to be sensitive to the value of tuning parameters. We leave an efficient implementation of this approach to future research.

The closed form solutions for $\tilde q_{al}$ and $\tilde h_{al}$ are 
\begin{align*}
    \tilde q_{al}(z, x) &= \sum_{i\in I_l^c} \gamma_{ali} K_{\mathcal Q_a}(Z_i, X_i;z, x),\\
    \text{and}\qquad \tilde h_{al}(w, x) &= \sum_{i\in I_l^c} \alpha_{ali} K_{\mathcal H_a}(W_i, X_i;w, x),
\end{align*}
where
\begin{align*}
    \alpha_{al} &= -(K_{\mathcal H_a, l}D_{al}\Gamma_{al}D_{al}K_{\mathcal H_a,l}+m_l^2\lambda_{\mathcal H_a}^{h_a} K_{\mathcal H_a, l})^+ K_{\mathcal H_a,l}D_{al}\Gamma_{al} g_{a2,l},\\
    \gamma_{al} &= -(K_{\mathcal Q_a, l}D_{al}\Omega_{al}D_{al}K_{\mathcal Q_a,l}+m_l^2\lambda_{\mathcal Q_a}^{q_a} K_{\mathcal Q_a, l})^+ K_{\mathcal Q_a,l}D_{al}\Gamma_{al} g_{a3,l},\\
    \Gamma_{al} &= \dfrac{1}{4}K_{\mathcal Q_a^*,l}(\dfrac{1}{m_l}K_{\mathcal Q_a^*,l} + \lambda_{\mathcal Q_a^*}^{h_a}I_{al})^{-1},\\
    \Omega_{al} &= \dfrac{1}{4}K_{\mathcal \mathcal H_a^*,l}(\dfrac{1}{m_l}K_{\mathcal H_a^*,l} + \lambda_{\mathcal H_a^*}^{q_a}I_{al})^{-1},
\end{align*}
$D_{al}$ is a $(n-n_l)\times (n-n_l)$ diagonal matrix with elements $\{g_{a1}(A_i, Y_i, X_i):i\in I_l^c\}$, $K_{\mathcal H,l}$ and $K_{\mathcal Q,l}$ are  respectively the $(n-n_l)\times (n-n_l)$ empirical kernel matrices with data in the sample $I_l^c$, $g_{a2,l}$ is the $(n-n_l)$-vector with elements $\{g_{a2}(Y_i, X_i):i\in I_l^c\}$, and $g_{a3,l}$ is the $(n-n_l)$-vector with elements $\{g_{a3}(A_i, Y_i, X_i):i\in I_l^c\}$.

\clearpage
\section{Regularity conditions in Theorem~\ref{thm:asym}}
\label{supp:regularity-conditions}

For a function class $\mathcal F$, let  $\mathcal F_B:=\{f\in\mathcal F:\lVert f\rVert_{\mathcal F}^2\leq B\}$.

Let $\{\mu^{h_a}_j\}_{j=1}^\infty$ and $\{\varphi_{j}^{h_a}\}_{j=1}^\infty$ be the eigenvalues and eigenfunctions of the RKHS $\mathcal H_a$ and $\{\mu^{q_a}_j\}_{j=1}^\infty$ and $\{\varphi_{j}^{q_a}\}_{j=1}^\infty$ be the eigenvalues and eigenfunctions of the RKHS $\mathcal Q_a$. For any $m\in \mathbb N_+$, let $V_m^{\mathcal H_a}$ and $V_m^{\mathcal Q_a}$ be the $m\times m$ matrices with entry $(i,j)$ defined respectively  as $$[V_m^{\mathcal H_a}]_{i,j}=E[E[\varphi_i^{\mathcal H_a}(W, X)beta_0\mid Z, X, S=1]E[\varphi_j^{\mathcal H_a}(W, X)\mid Z, X, S=1]]$$ and $$[V_m^{\mathcal Q_a}]_{i,j}=E[E[\varphi_i^{\mathcal Q_a}(Z, X)\mid W, X, S=1]E[\varphi_j^{\mathcal Q_a}(Z, X)\mid W, X, S=1]].$$ Let $\lambda_m^{\mathcal H_a}$ and $\lambda_m^{\mathcal Q_a}$ be the minimum eigenvalues of $V_m^{\mathcal H_a}$ and $V_m^{\mathcal Q_a}$, respectively.

\begin{enumerate}[R.1]
    \item The functions $g_{a1}(A, Y, X)=c(X)I(A=a)(1-Y)$, $g_{a2}(A, Y, X)=c(X)I(A=a)Y$ and $g_{a3}(Y, X)=c(X)(1-Y)$ are bounded.
    \item There exists a constant $\sigma_1$ such that $$\mid [E[c(X)I(A=a)(1-Y)\mid Z, W, X, S=1]\mid >\sigma_1>0$$
    almost surely for $a=0, 1$.
    \item The nuisance functions $q(A, Z, X)$ and $h(A, W, X)$ have finite second moments.
    \item The nuisance functions $q_a(Z, X)$ and their estimators $\tilde q_{al}(Z, X)$ and $\tilde h_{al}(W, X)$ in each fold of cross-fitting satisfy
    $$\min\left\{\sup_{z, x}\mid \tilde q_{al}(z, x) + \sup_{w, x} h_a(w, x)\mid , \sup_{z, x}\mid q_{a}(z, x) + \sup_{w, x} \tilde h_{al}(w, x)\mid \right\}<\infty.$$
    \item $\lVert \tilde  q_{al}-q_{a}\rVert_2 = o_p(1)$, $\lVert \tilde h_{al}-h_a\rVert_2 = o_p(1)$;
    \item There exists a constant $C$ such that $\tilde h_a\in\mathcal H_{aC}$ and $\tilde q_a\in\mathcal Q_{aC}$ for $a\in\{0,1\}$.
    \item $$\mid  E[E[\varphi_i^{\mathcal H_a}(W, X)\mid Z, X, S=1]E[\varphi_j^{\mathcal H_a}(W, X)\mid Z, X, S=1]] \mid \leq c^{\mathcal H_a}\lambda_m^{\mathcal H_a}$$ and $$\mid E[E[\varphi_i^{\mathcal Q_a}(Z, X)\mid W, X, S=1]E[\varphi_j^{\mathcal Q_a}(Z, X)\mid W, X, S=1]]\mid \leq c^{\mathcal Q_a}\lambda_m^{\mathcal Q_a}$$ for some constant $c^{\mathcal H_a}$ and $c^{\mathcal Q_a}$.
    \item $$\lVert E[I(A=a)c(X)(1-Y)\{\tilde q_{al}(Z, X)-q_a(Z, X)\rVert W, X, S=1\}]\rVert_2 = O(n^{-{r_q}})$$ and $$\lVert E[I(A=a)c(X)(1-Y)\{\tilde h_{al}(W, X)-h_a(W, X)\rVert Z, X, S=1\}]\rVert_2 = O(n^{-r_h}).$$
    \item $\mu_m^{\mathcal H_a}\sim m^{-s_{h}}$, $\mu_m^{\mathcal Q_a}\sim m^{-s_{q}}$, $\lambda_m^{\mathcal H_a}\sim m^{-t_{h}}$, and $\lambda_m^{\mathcal Q_a}\sim m^{-t_{q}}$, where the decay rates $s_h$, $s_q$, $t_h$ and $t_q$ satisfies
    $$\max\{r_h + \dfrac{s_hr_q}{s_h + t_h}, r_q + \dfrac{s_qr_h}{s_q + t_q}\}>\dfrac{1}{2}.$$
\end{enumerate}

\clearpage 

\section{Proof of Theorem~\ref{thm:asym}}\label{sec:proof-rkhs}
Under Model~\ref{mod:logit-model} and Assumptions~\ref{assump:selection}, \ref{assump:nc} and \ref{assump:completeness-a}, Theorem~\ref{thm:PIPW}(b) states that $\beta_0 = \log(\zeta_{1c}/\zeta_{0c})$. By delta-method, it suffices to prove that $\tilde \zeta_{ac}$ is asymptotically linear with the influence function $IF_{\zeta_{ac}}(O)$, the proof of which follows \citet{ghassami2021minimax} Theorem 2 and Lemma 1, which we reproduced below.

For $a=0,1$ and a given value $\delta >0$, we define
$\mathcal H_{aC}^{|\delta}:= \{h_a\in\mathcal H_{aC}:\lVert E[h_a(W, X)\mid Z, X, S=1]\rVert_2\leq \delta\}$ and $\mathcal Q_{aC}^{|\delta}:= \{q_a\in\mathcal Q_{aC}:\lVert E[q_a(Z, X)\mid W, X, S=1]\rVert_2\leq \delta\}$. We define the measures of ill-posedness $\tau_{h_a}(\delta):=\sup_{h_a\in\mathcal H_{aC}^{|\delta}}\lVert h_a\rVert_2$ and  $\tau_{q_a}(\delta):=\sup_{q_a\in\mathcal Q_{aC}^{|\delta}}\lVert q_a\rVert_2$.

\begin{lemma}\label{lemma:ghassami-lemma}(\citet{ghassami2021minimax} Lemma 1)
Under regularity conditions R.6 and R.7, for any $\epsilon>0$ we have
$$\tau_{h_a}(\epsilon)^2\leq \min_{m\in\mathbb N_+}\left\{\dfrac{4\epsilon^2}{\lambda_m^{\mathcal H_a}} +[4 (c^{\mathcal H_a})^2 + 1]C\mu_{m+1}^{\mathcal H^a}\right\}$$
and $$\tau_{q_a}(\epsilon)^2\leq \min_{m\in\mathbb N_+}\left\{\dfrac{4\epsilon^2}{\lambda_m^{\mathcal Q_a}} +[4 (c^{\mathcal Q_a})^2 + 1]C\mu_{m+1}^{\mathcal Q^a}\right\}.$$
\end{lemma}

By Lemma~\ref{lemma:ghassami-lemma} and R.9, we have that
$$\tau_{h_a}(\epsilon)=O\left(\epsilon^{\dfrac{s_h}{s_h + t_h}}\right)\qquad\text{and}\qquad\tau_{q_a}(\epsilon)=O\left(\epsilon^{\dfrac{s_q}{s_q + t_q}}\right)$$ The condition R.9 further implies that
\begin{align}
    \min\{n^{-r_h}\tau_{h_a}(n^{-r_q}), n^{-r_q}\tau_{q_a}(n)\} &= O\left(\min\left\{n^{-\left(r_h + \dfrac{s_hr_q}{s_h + t_h}\right)}, n^{-\left(r_q + \dfrac{s_qr_h}{s_q + t_q}\right)}\right\}\right)\nonumber\\&=o(n^{-1/2}).\label{eq:resid-convergence}
\end{align}

The result follows by the theorem below:
\begin{theorem}
(\citet{ghassami2021minimax} Theorem 2) Under regularity conditions R.1-R.6 and R.8, if Equation~\eqref{eq:resid-convergence} holds in addition, then 
the estimator $\tilde\theta_{ac}$ is asymptotically linear with the influence function
$$IF_{\zeta_{ac}}(O).$$
\end{theorem}

\clearpage

\section{Details of the simulation studies}\label{supp:sim}
We generate the data for a target population of $N$ observations according to the following five scenarios
\subsection*{Scenario I: Univariate binary $U$, $Z$ and $W$}

\begin{align*}
    U & \sim \text{Bernoulli}(0, 0.5);\\
    A\mid U & \sim \text{Bernoulli}(\expit(0.2 + 0.4U))\\
    Z\mid A, U & \sim \text{Bernoulli}(0.2 + 0.1 A + 0.4U + 0.2AU);\\
    W\mid U & \sim \text{Bernoulli}(0.2 + 0.4U);\\
    Y\mid A, U & \sim \text{Bernoulli}(\expit(-0.405 -1.609A - 0.7U))\\
    S\mid U, A, Y & \sim \text{Bernoulli}(\exp(-1.7 + 0.2A + 0.4Y + 0.7U)).
\end{align*}

\subsection*{Scenario II-V: Continuous $U$, $X$, $Z$, $W$.}

\begin{align*}
    U & \sim \text{Uniform}(0, 1);\\
    X & \sim \text{Uniform}(0,1); \\
    A\mid U, X & \sim \text{Bernoulli}(\expit(-1 + U + X))\\
    Z\mid A, U, X & \sim \text{N}(0.25A + 0.25X + 4U, 0.25^2);\\
    W\mid Y, U, X & \sim \text{N}(0.25Y + 0.25X + 4U, 0.25^2);\\
    Y\mid A, U, X & \sim \text{Bernoulli}(\expit(-3.89 + \beta_0A - 2U - X))\\
    S\mid U, X, A, Y & \sim \text{Bernoulli}(\exp(-5 + 0.4A + 4Y + 0.3U + 0.2X).
\end{align*}

\clearpage
\section{Estimating conditional odds ratio under effect modification by $X$}\label{supp:eff-mod}

In Section~\ref{sec:homo-or}, we assumed that the effect size of odds ratio is constant across every $(U, X)$ strata. However, in practical situation, the treatment effect may be heterogeneous. Therefore, we relax Model~\ref{mod:logit-model} and propose Model~\ref{mod:logit-model-x} below that allows heterogeneous odds ratio across $X$ strata:
\begin{model}[Heterogeneous odds ratio model by $X$]\label{mod:logit-model-x}
\begin{equation}
    \logit \,P(Y=1\mid A, U, X) = \beta_0(X) A + \eta(U, X).
\end{equation}
\end{model}

In Model~\ref{mod:logit-model-x}, the odds ratio treatment effect $\beta_0(X)$ is homogeneous with respect to the latent factors $U$ but is allowed to vary by $X$. In such scenarios, if the negative control variables are available, the odds ratio function $\beta_0(X)$ can still be identified, similar to Theorems~\ref{thm:PIPW}, \ref{thm:POR} and \ref{thm:PDR}, as stated in the following theorem:

\begin{theorem}\label{thm:identification-beta-x}Under Model \ref{mod:logit-model-x} and Assumptions~\ref{assump:selection}, \ref{assump:nc}, \ref{assump:trt-bridge} and \ref{assump:outcome-bridge}, then the odds ratio function $\beta_0(X)$ satisfies the following moment equations:
\begin{align*}
    &E\left[c(X)I(A=a)Yq(A, Z, X)e^{-\beta_0(X)}\mid S=1\right]=0;\\
    &E\left[c(X)(1-Y)\left\{h(1, W, X) - h(0, W, X)e^{\beta_0(X)}\right\}\mid S=1\right]=0;\\
    &E\bigg\{c(X)\bigg[(-1)^{1-A}q(A, Z, X)e^{-\beta_0(X) A}\left\{Y - (1-Y)h(A, W, X)\right\}+\nonumber\\
        &\qquad \qquad (1-Y)\left\{h(1, W, X)e^{-\beta_0(X)} - h(0, W, X)\right\}\bigg]\mid S=1\bigg\}=0.
\end{align*}
\end{theorem}

In cases where a parametric model for the odds ratio function $\beta_0(X;\alpha)$ is considered appropriate with a finite dimensional parameter $\alpha$, similar to the previous PIPW, POR and PDR estimation, the low-dimension parameter $\alpha$ can be estimated by solving the estimating equations

\begin{align}
    &\dfrac{1}{n}\sum_{i=1}^n {(-1)}^{1-A_i}c(X_i)q(A_i, Z_i, X_i; \tau)Y_ie^{-\beta_0(X;\alpha) A_i}=0;\label{eq:est-PIPW-x}\\
    &\dfrac{1}{n}\sum_{i=1}^n c(X_i)(1-Y_i)\left\{h(1, W_i, X_i;\hat\psi)e^{-\beta_0(X_i;\alpha)} - h(0, W_i, X_i;\hat\psi)\right\}=0;\label{eq:est-por-x}\\
    \text{or}\qquad &\dfrac{1}{n}\sum_{i=1}^n c(X_i)\bigg[(-1)^{1-A_i}q(A_i, Z_i, X_i)e^{-\beta_0(X;\alpha) A}\left\{Y_i - (1-Y_i)h(A_i, W_i, X_i)\right\}+\nonumber\\
        &\qquad \qquad (1-Y_i)\left\{h(1, W_i, X_i)e^{-\beta_0(X;\alpha)} - h(0, W_i, X_i)\right\}\bigg]=0,\label{eq:est-pdr-x}
\end{align}where $\hat\tau$ and $\hat\psi$ are estimators for the parameters indexing the treatment and outcome confounding bridge functions, obtained by solving~Equations~\eqref{eq:q-estimation} and \eqref{eq:h-estimation}, respectively, and $c(X)$ is a user-specified function with dimension no smaller than that of $\alpha$.

Alternatively, a similar kernel method to the one in Section~\ref{supp:RKHS} of the Supplementary Materials can be used to obtain a cross-fitting minimax learning estimator of $\alpha$ which solves Equation~\eqref{eq:est-pdr-x} and nonparametrically estimates the nuisance functions $q$ and $h$. With an estimator $\hat\alpha$ using either of the above methods, the conditional odds ratio at $X=x$ can then be estimated by $\beta_0(x;\hat\alpha)$.

Nonparametric estimation of $\beta_0(X)$ is possible, for example, by using the minimax learning approach \citep{dikkala2020minimax} or the semi-nonparametric sieve generalised method of moments approach~\citep{ai2003efficient,chen2007large}. Such a nonparametric estimator may be of interest, for example, to more flexibly account for the heterogeneity of treatment effects and making flexible prediction on the ``individual treatment effect''.

\clearpage
\section{Extension to polytomous treatment and effect}\label{supp:polytomous}

Suppose we want to study the effect of a polytomous treatment $A$ on a polytomous outcome $Y$, where $A$ have levels $a_0,a_1,\dots,a_J$ and $Y$ have levels $y_0,\dots, y_K$. Here $a_0$ and $y_0$ denote the reference treatment and outcome, respectively. We make the following homogeneity assumption:
\begin{model}[Homogeneous odds ratio model]\label{mod:logit-mod-polytomous}
    \begin{align*}
    &\log\left(\dfrac{P(Y\mid A, U, X)P(Y=y_0\mid A=a_0, U, X)}{P(Y\mid A=a_0, U, X)P(Y=y_0\mid A, U, X)}\right)=\\ &\qquad\qquad\sum_{\substack{j\in\{1,\dots,J\}\\k\in\{1,\dots, K\}}}\beta_j^k I(A=a_j, Y=y_k)+\eta^k(U,X).\end{align*}
\end{model}
Model~\ref{mod:logit-mod-polytomous} indicates that the odds ratio of $Y=y_k$ relative to $Y=y_0$ against the treatment $A=a_j$ versus $A=a_0$ is $\beta_j^k$ across all strata of $(U,X)$. Model~\ref{mod:logit-mod-polytomous} is a natural semiparametric extension to the polytomous logistic regression model~\citep{engel1988polytomous}. When $J=1$ and $K=1$, Model~\ref{mod:logit-mod-polytomous} reduces to Model~\ref{mod:logit-model}. The goal is to estimate and make inference on every $\beta_j^k$ in the presence of latent factors $U$ and confounded outcome-dependent sampling.

We make the following assumptions, corresponding to Assumption~\ref{assump:selection}, \ref{assump:trt-bridge} and~\ref{assump:outcome-bridge}:
\begin{assumption}[No effect modification by $A$ on the outcome-dependent selection]\label{assump:selection-polytomous}For $a=a_0,\dots,a_J$ and real-valued unknown functions $\xi^1(U,X),\dots,\xi^K(U,X)$, the sampling mechanism satisfies
\begin{align*}
    \dfrac{P(S=1\mid Y=y_k, A=a, U, X)}{P(S=1\mid Y=y_0, A=a, U, X)}=\xi^k(U,X)
\end{align*}
for $k=1,\dots, K$.
\end{assumption}

\begin{assumption}[Confounding bridge functions]\label{assump:bridge-fcts-polytomous}
The exists a treatment confounding bridge function $q(A, Z, X)$ and $K$ outcome confounding bridge functions $h_k(A, W, X)$ such that for $k=1,\dots, K$ and $a=a_0,a_1,\dots, a_J$,
\begin{align}
    &E\{q(a, Z, X)\mid A=a, U, X\}=1/P(A=a\mid U, X, Y=y_0,S=1)\label{eq:trt-bridge-u-polytomous}\\
    &E\{I(Y=y_0)h_k(a, W, X)\mid A=a,U, X, S=1\}=P(Y=y_k\mid A=a, U, X, S=1).\label{eq:outcome-bridge-u-polytomous}
\end{align}

\end{assumption}

We state the following results for inference of $\beta_j^k$. The proofs are similar to those in Section~\ref{supp:proof} and are therefore omitted.
\begin{lemma}\label{lemma:ident-or-polytomous}
Under Model~\ref{mod:logit-mod-polytomous} and Assumption~\ref{assump:selection-polytomous}, the log odds ratio $\beta_j^k$ satisfies
$$\beta_j^k = \log(\theta_{jc}^k/\theta_{0c}^k),$$
where 
$$\theta_{jc}^k = E\left\{c(X)\dfrac{P(A=a_j\mid U, X, Y=y_k,S=1)}{P(A=a_j\mid U, X, Y=y_0, S=1)}P(Y=y_k\mid U, X, S=1)\mid S=1\right\}.$$
and $c(X)$ is an arbitrary real-valued function that satisfies $\theta_{jc}^k\neq 0$ for $j=0,1,\dots, K$.
\end{lemma}

\begin{theorem}[Identification of $q(A, Z, X)$ and $h_k(A, W, X)$]\label{thm:q-h-ident-polytomous}

The functions $q(\cdot)$ and $h_k(\cdot)$, $k=1,\dots, K$, satisfy the moment conditions 
    \begin{align*}
        &E\left[I(Y=y_0)\left\{\kappa^*_1(A, W, X)q(A, Z, X)-\sum_{j=0}^J\kappa^*_1(a_j, W, X)\right\}\mid S=1\right]=0\\
        &E\left[\kappa^*_{2k}(A, Z, X)\left\{I(Y=y_0)h_k(A, W, X) - I(Y=y_k)\right\}\mid S=1\right]=0,\qquad \mbox{k=1,\dots, K},
    \end{align*}
    where $\kappa^*_1(\cdot)$ and $\kappa_{2k}$ are arbitrary functions.
\end{theorem}

\begin{theorem}[Identification of $\theta_{jc}^k$]\label{thm:tac-ident-polytomous}
Assume Model~\ref{mod:logit-mod-polytomous} and Assumptions~\ref{assump:selection-polytomous} and \ref{assump:nc} hold, then
\begin{enumerate}[(a)]
    \item if there exists a function $q$ that solves Equation~\eqref{eq:trt-bridge-u-polytomous}, then
    $$\theta_{jc}^k =E\left\{I(A=a_j, Y=y_k)c(X)q(a_j, Z, X)\mid S=1\right\};$$
    
    \item if there exists functions $h^k$, $k=1,\dots, K$ that solve Equation~\eqref{eq:outcome-bridge-u-polytomous}, then
    $$\theta_{jc}^k =E\left\{c(X)I(Y=y_0)h_k(a_j, W, X)\mid S=1\right\};$$
    
    \item if either (i) the function $q$ solves Equation~\eqref{eq:trt-bridge-u-polytomous}, or (ii) the functions $h_k$'s solve Equation~\eqref{eq:outcome-bridge-u-polytomous}, then
\begin{align*}\theta_{jc}^k &= E\big\{c(X)\big[I(Y=y_0)h_k(a_j, W, X)- \nonumber \\&\qquad I(A=a_j)q(a_j, Z, X)\big\{I(Y=y_0)h_k(a_j, W, X)-I(Y=y_k)\big\}\big]\mid S=1\big\}.\end{align*}
    
\end{enumerate}

\end{theorem}

By Theorem~\ref{thm:q-h-ident-polytomous} and \ref{thm:tac-ident-polytomous}, to estimate $\beta_j^k$ one may:
\begin{enumerate}
    \item Specify suitable parametric working models $q(A, Z, X;\tau)$ and $h_k(A, W, X;\psi_k)$;
    \item Estimate the nuisance parameters $\tau$ and $\psi$ by solving
    \begin{align*}
        &\sum_{i=1}^n I(Y_i = y_0)\left\{\kappa^*_1(A_i, W_i, X_i)q(A_i, Z_i, X_i;\tau) - \sum_{j=0}^J \kappa^*_1(a_j, W_i, X_i)\right\}=0;\\
        &\sum_{i=1}^n  \kappa^*_{2k}(A_i, Z_i, X_i)\{I(Y_i = y_0)h_k(A_i, W_i, X_i;\psi_k) - I(Y_i = y_k)\}=0\\
    \end{align*}for $k=1,2,\dots, K$, where $\kappa^*_1$ and $\kappa^*_{2k}$ are user-specified functions that dimensions at least as large as that of $\tau$ and $\psi_k$ respectively. Denote the resulting estimator as $\hat\tau$ and $\hat\psi_k$.
    \item Estimate $\theta_{jc}^k$ by
    \begin{align*}
        &\hat\theta_{jc,\text{PIPW}}^k = \dfrac{1}{n}\sum_{i=1}^n I(A=a_j, Y=y_k)c(X_i)q(a_j, Z_i,X_i;\hat\tau);\\
        &\hat\theta_{jc,\text{POR}}^k =\dfrac{1}{n}\sum_{i=1}^n c(X_i)I(Y_i = y_0)h_k(a_j, W_i, X_i;\hat\psi_k);\\
        \text{or}\qquad &\hat\theta_{jc,\text{PDR}}^k = \dfrac{1}{n}\sum_{i=1}^n c(X_i)\bigg[I(Y_i = y_0)h_k(a_j, W_i, X_i;\hat\psi_k)-\\&\qquad I(A_i = a_j)q(a_j, Z_i, X_i;\hat\tau)\left\{I(Y_i = y_0)h_k(a_j, W_i, X_i;\hat\psi_k) - I(Y_i = y_k)\right\}\bigg]
    \end{align*}
    for $j=0,1,\dots, J$ and $k=1,\dots, K$ and a user-specified real-valued function $c(X)$ (one may simply set $c(X)=1$).
    
    \item The resulting estimators for $\beta_j^K$ include the PIPW estimator
    $\hat\beta_{j, \text{PIPW}}^k=\log(\hat\theta_{jc,\text{PIPW}}^k / \hat\theta_{0c,\text{PIPW}}^k)$, the POR estimator $\hat\beta_{j, \text{POR}}^k=\log(\hat\theta_{jc,\text{POR}}^k / \hat\theta_{0c,\text{POR}}^k)$, and the PDR estimator $\hat\beta_{jc, \text{PDR}}^k=\log(\hat\theta_{jc,\text{PDR}}^k / \hat\theta_{0c,\text{PDR}}^k).$
\end{enumerate}

Under Model~\ref{mod:logit-mod-polytomous} and Assumptions~\ref{assump:nc} and \ref{assump:selection-polytomous}, the PDR estimator satisfies doubly robustness, i.e. $\hat\beta_{jc,\text{POR}}^k\overset{p}{\rightarrow}\beta_j^k$ if either (i) the parametric working model $q(A, Z, X;\tau)$ is correctly specified, and the completeness assumption~\ref{assump:completeness-a}(b) holds, or (ii) the parametric working model $h(A, W, X;\psi)$ is correctly specified, and Assumption~\ref{assump:completeness-b}(b) holds.

The kernel learning approach in Section~\ref{supp:RKHS} can similarly be employed to obtain semiparametric estimators for $\beta_j^k$ with flexible modeling for $q$ and $h_k$.

\end{document}